\documentclass[journal,twoside,web]{ieeecolor} 
\usepackage{generic}
\usepackage{graphicx}
\usepackage{amsmath} 
\usepackage{amssymb}
\usepackage{amsfonts}
\usepackage{mathrsfs}
\usepackage{mathtools}
\usepackage{cite}
\usepackage[bbgreekl]{mathbbol}
\usepackage{multirow}
\usepackage{makecell}
\usepackage{booktabs}
\usepackage{xpatch}
\usepackage{float}
\usepackage{tikz}
\usepackage{pgfplots} 
\usepackage{pgfgantt}
\usepackage{pdflscape}

\pgfplotsset{compat=newest}

\let\labelindent\relax
\usepackage{enumitem}

\usepackage{amsthm}

\DeclareSymbolFontAlphabet{\mathbb}{AMSb}
\DeclareSymbolFontAlphabet{\mathbbl}{bbold}

\DeclareMathOperator{\diag}{diag}
\DeclareMathOperator{\rank}{rank}
\DeclareMathOperator{\blkdiag}{blkdiag}

\DeclareMathOperator{\diff}{d}

\DeclareMathOperator*{\argmin}{argmin}

\DeclareMathOperator{\Ima}{Im}
\DeclareMathOperator{\vspan}{span}

\newcommand{\mc}{\mathcal}
\newcommand{\ddt}{\tfrac{\diff}{\diff \!t}}
\newcommand{\norm}[1]{\left \lVert #1 \right \rVert}

\makeatletter
\xpatchcmd{\@thm}{\thm@headpunct{.}}{\thm@headpunct{}}{}{}
\makeatother

\newcommand\blue[1]{{\color{blue}#1}}
\definecolor{lightgray}{gray}{0.9}

\newtheorem{theorem}{Theorem}
\newtheorem{lemma}{Lemma}
\newtheorem{proposition}{Proposition}
\newtheorem{assumption}{Assumption}
\newtheorem{remark}{Remark}
\newtheorem{definition}{Definition}
\newtheorem{condition}{Condition}

\newtheorem{corollary}{Corollary}
\newtheorem{property}{Property}

\title{\LARGE \bf  Networked dynamics with application to frequency stability of grid-forming power-limiting droop control}

\author{Amirhossein Iraniparast  and Dominic Gro\ss{} \thanks{This work was supported in part by the National Science Foundation under Grant No. 2143188. A. Iraniparast and D. Gro\ss{} are with the Department of Electrical and Computer Engineering at the University of Wisconsin-Madison, USA; e-mail: iraniparast@wisc.edu, dominic.gross@wisc.edu}}

\begin{document}
\maketitle

\begin{abstract}
    In this paper, we study a constrained network flow problem and associated networked dynamics that resemble but are distinct from the well-known primal-dual dynamics of the constrained flow problem. Crucially, under a change of coordinates, the networked dynamics coincide with primal-dual dynamics associated with the constrained flow problem in edge coordinates. Next, we show that, under mild feasibility assumptions, the networked dynamics are globally asymptotically stable with respect to the set of optimizers of its associated constrained flow problem in nodal coordinates. Subsequently, we apply our stability results to establish frequency stability of power-limiting grid-forming droop control. Compared to conventional grid-forming droop control, power-limiting droop control explicitly accounts for active power limits of the generation (e.g., renewables) interfaced by the converter. While power-limiting droop control has been demonstrated to work well in simulation and experiment, analytical results are not readily available.  
    Moreover, we (i) show that the converter frequencies synchronize to a common synchronous frequency for each grid-forming converter, (ii) characterize the synchronous frequency in the case of converters operating at their power limit, and (iii) establish that power-limiting droop control exhibits power-sharing properties similar to conventional unconstrained droop control. Finally, the analytical results are illustrated using an Electromagnetic transient (EMT) simulation.
\end{abstract}
\begin{keywords}
Grid-forming control, frequency synchronization, power limiting
\end{keywords}
\section{Introduction}
The transition from bulk power generation towards decarbonization and renewable energy sources results in significantly different power system dynamics. In particular, renewable generation is interfaced through power electronic devices that significantly differ from conventional synchronous generators in terms of their response time, device constraints (e.g., power and current limits), and dynamic interactions through the grid. As a result, introducing renewable generation into large-scale power system challenges standard operating and control paradigms and jeopardizes system stability~\cite{KJZ+2017,MDHHV2018}. For instance, the constraints of power converters and renewable generation sources such as power limits need to be considered in the stability analysis of emerging power systems.

Continuous-time primal-dual gradient descent dynamics~\cite{ashish} have been used to analyze existing system-level control algorithms such as automatic generation control in multi-machine power systems~\cite{LZC2016} and design distributed power flow control~\cite{CDA2020}. These works focus on equality constrained optimization problems that arise from secondary and tertiary control of power systems. In contrast, applying primal-dual dynamics to inequality constrained control of converters results in dynamics that (i) cannot be implemented using only local information, and (ii) do not match existing controls such as power limiting droop control~\cite{DLK2019}. While power limiting droop control closely resembles primal-dual dynamics of a constrained power flow problem, connections between the two have not been investigated. This work investigates a constrained network flow problem and associated networked dynamics that can be implemented using only local information. We show that, under a suitable change of coordinates, the networked dynamics coincide with the well-known primal-dual dynamics of the constrained flow problem. Moreover, in their original coordinates, the networked dynamics coincide with the frequency dynamics of a power system containing converters using grid-forming power-limiting droop control.

Today, most renewables are interfaced by dc/ac voltage source converters using so-called grid-following control that require a stable and slowly changing ac voltage (i.e., magnitude and frequency) and jeopardize grid stability when disturbances occur~\cite{KH2024}. However, because grid-following control explicitly controls the converter current, incorporating power limits is straightforward. In contrast, grid-forming converters impose stable and self-synchronizing ac voltage dynamics at their grid terminals and are commonly envisioned to replace synchronous machines as the cornerstone of future power systems. Prevalent grid-forming controls include droop control~\cite{CDA1993}, virtual synchronous machine control (VSM)~\cite{DSF2015}, and dispatchable virtual oscillator control (dVOC) ~\cite{GCBD2019}. However, device constraints remain a significant concern and existing stability results~\cite{DB2012, SOARS2014, SGRS2013, DS2014,SGCD2021} do not account for constraints. While the majority of theoretical works on grid-forming controls neglect constraints, current limiting in grid-forming controls has received significant attention in the application oriented literature~\cite{SCG2010, PD2015, FLZ2022, BCL+2024} and only few works investigate dc voltage limits~\cite{CPLJ2016} and power limits~\cite{DLK2019}. Power-limiting droop control  combines conventional droop control with proportional-integral controls that activate when the converter reaches its power limit~\cite{DLK2019}. While power-limiting droop control has been demonstrated to work well in practice~\cite{DLK2019}, to the best of the authors' knowledge, no analytical stability conditions or theoretical results are available in the literature.

To this end, we first formulate a generic constrained network flow problem and associated projected networked nodal dynamics that, in contrast to primal-dual dynamics of the constrained flow problem, can be implemented using only local flow measurements. Our main contribution is to show that Carath\'eodory solutions of the projected networked dynamics are asymptotically stable with respect to Karush-Kuhn-Tucker (KKT) points of the constrained flow problem. To obtain this result, we show that the networked dynamics correspond to primal-dual dynamics of the constrained network flow problem in so-called edge coordinates~\cite{ZM2011}. Moreover, we characterize the active constraint set and synchronization of nodal dynamics.

To apply the results in the context of power limiting droop control, we first reformulate the frequency dynamics of a multi-converter power system with converters using grid-forming power-limiting droop control as a projected dynamical system. We then apply our stability result to establish that the multi-converter power system is asymptotically stable with respect to the set of KKT points of a constrained optimal dc power flow problem. Moreover, we establish that the grid-forming converters synchronize to a common synchronous frequency. Next, we formally characterize the relationship between the overall load and the active constraint set and synchronous frequency upon convergence of the power system to a KKT point of the constrained power flow problem.

Moreover, we establish that, upon  convergence, the converters exhibit properties similar to so-called power-sharing in  unconstrained droop control~\cite{SDB2013}. Specifically, power-limiting droop control results in power-sharing up to the power limit, i.e., converters share the additional load according to their droop coefficient until reaching their power limit. This result also establishes that (i) power-limiting droop control cannot converge to operating points at which some converters are at their upper power limit while other converters are at their lower power limit, and that (ii) the synchronous frequency is a function of the overall load, droop coefficients of converters that are not at their power limit, and power limits of converters that are operating at their power limit. Overall, these results are important from a practical point of view to establish that power-limiting droop control does not converge to counter-intuitive operating points that are not aligned with assumptions of higher-level power system controls and operation. Finally, the reduced-order power system model and analytical results are validated and illustrated using an Electromagnetic transient (EMT) simulation of the IEEE 9-bus system.

This paper is organized as follows. Sec.~\ref{sec:model} introduces the network model, control objectives, and networked dynamics. Next, Sec.~\ref{sec:stabanalysis} defines a constrained optimization problem in edge coordinates to establish stability of the networked dynamics. Sec.~\ref{sec:application} applies the results to establish frequency stability of multi-converter power systems using grid-forming power-limiting droop control. A numerical case study to validate the main results is provided in Sec.~\ref{sec:numerical}. Finally, Sec.~\ref{sec:conclusion} provides conclusions and topics for future work.

\subsection*{Notation}
We use $\mathbb{R}$ and $\mathbb N$ to denote the set of real and natural numbers and define, e.g., $\mathbb{R}_{\geq 0}\coloneqq \{x \in \mathbb R \vert x \geq 0\}$. Moreover, we use $\mathbb{S}_{\succ 0}^n$ and $\mathbb{S}_{\succeq 0}^n$ to denote the set of real positive definite and positive semidefinite matrices. For column vectors $x\in\mathbb{R}^n$ and $y\in\mathbb{R}^m$ we define $(x,y) = [x^\mathsf{T}, y^\mathsf{T}]^\mathsf{T} \in \mathbb{R}^{n+m}$. Furthermore, $I_n$, $\mathbbl{0}_{n\times m}$, $\mathbbl{0}_{n}$, and $\mathbbl{1}_n$ denote the $n$-dimensional identity matrix, $n \times m$ zero matrix, and column vectors of zeros and ones of length $n$ respectively.  Moreover, we use $\norm{x}_{Q}=\sqrt{x^\mathsf{T} Q x}$ to denote the weighted Euclidean norm and $\norm{x}_{\mc C} \coloneqq \min_{z\in\mc C} \norm{z-x}$ denotes the point-to-set distance from $x$ to $\mc C$. The cardinality of a discrete set $\mc X$ is denoted by $|\mc X|$. The Kronecker product is denoted by $\otimes$. We use $\varphi_x(t,x_0)$ to denote a (Carath\'eodory) solution of $\ddt x = f(x)$ at time $t \in \mathbb{R}_{\geq 0}$ starting from $x_0$ at time $t=0$.

\section{Network model and control}\label{sec:model}
    In this section, we introduce the network model and control that will be considered throughout the paper.
    \subsection{Network and model}
    Consider a network modeled by a simple, connected and undirected graph $\mathcal{G}\coloneqq \{\mathcal{N}, \mathcal{E}, \mathcal{W}\}$ with set $\mathcal{E}\coloneqq\mathcal{N}\times\mathcal{N}$ corresponding to $|\mathcal{E}|=e$ edges, set $\mathcal{N}$ corresponding to $|\mathcal{N}|=n$ nodes, and set of edge weights $\mathcal{W}=\left\{w_1, \ldots, w_e\right\}$ with $w_i \in \mathbb{R}_{>0}$ for all $i \in \{1,\ldots,e\}$. Each node $i \in \mathcal{N}$ is associated with a local state variable $\theta_i \in \mathbb{R}$, constant unknown disturbance input $P_{L,i} \in \mathbb{R}$, and network injections $P_i \in \mathbb{R}$ that corresponds to interactions of the nodes (e.g., power flow in a power system). Considering the aforementioned definitions, the coupling through the network is modeled by
    \begin{align}\label{eq:dcpfeq}
        P \coloneqq L\theta + P_{L}, 
    \end{align}
    where $L\coloneqq BWB^\mathsf{T}$ is the Laplacian matrix of the undirected graph $\mc G$, $B \in {-1,0,1}^{n \times e}$ denotes the oriented incidence matrix of $\mc G$~\cite{LNS}, and $W=\diag\{w_i\}_{i=1}^{e}$. Moreover, $\theta = \left(\theta_1, \ldots, \theta_n\right) \in \mathbb{R}^n$ is the vector of nodal state variables, $P_{L} \coloneqq \left(P_{L,1}, \ldots, P_{L, n}\right) \in \mathbb{R}^n$ is the vector of unknown and constant disturbances, and $P = \left(P_1, \ldots, P_n\right) \in \mathbb{R}^n$ is the vector of network injections (i.e., interaction variables).

    \subsection{Objective and preliminary results}
    Our objective is to find nodal variables $\theta \in \mathbb{R}^n$ using only local information (i.e., the network injections $P_i$ and local states) that minimize the constrained flow problem (CFP)
    \begin{subequations}\label{eq:deCFP}
        \begin{align}
            &\min_{\theta, P} \quad \tfrac{1}{2} \|P-P^\star\|^2_{M} \label{eq:deCFP:dcpf:obj}\\
            & \text {s.t.} \quad  P_\ell \leq P \leq P_{u} \label{eq:deCFP:dcpf:lim}\\
            & \phantom{\text{s.t.}} \quad P = L \theta + P_{L} \label{eq:deCFP:dcpf}
            \end{align}
    \end{subequations}
    where $M \in \diag\{m_i\}_{i=1}^n \in \mathbb{S}^n_{\succ 0}$ is a diagonal matrix of weights, $P^\star \in \mathbb{R}^n$ is a vector of local references, and $P _{\ell} \in \mathbb{R}^n$ and $P_{u} \in \mathbb{R}^n$ model limits on the network injections. The following assumptions ensure feasibility of \eqref{eq:deCFP}.
\begin{assumption}[\textbf{Feasible injection limits and disturbances}]\label{assum:feas}
For all $i \in \mathcal{N}$, the limits $P_{\ell,i} \in \mathbb{R}^n$ and $P_{u,i} \in \mathbb{R}^n$ satisfy $P_{\ell,i} < P_{u,i}$. Moreover, the disturbance input $P_L \in \mathbb{R}^n$ satisfies $\sum_{i=1}^n P_{\ell,i} < \sum_{i=1}^n  P_{L,i} < \sum_{i=1}^n P_{u,i}$.
\end{assumption}
\begin{assumption}[\textbf{Feasible references}]\label{assum:setpoint}
    The setpoints $P^\star_i \in \mathbb{R}^n$ satisfy $P_{\ell,i} < P^\star_i < P_{u,i}$.
\end{assumption}
We require the following preliminary result that ensures feasibility of the CFP \eqref{eq:deCFP} under Assumption~\ref{assum:feas}.
\begin{proposition}[\textbf{Feasibility in nodal coordinates}]\label{prop:feas}
    There exists $\theta \in \mathbb{R}^n$ such that $P_{\ell} < L\theta + P_{L} < P_{u}$ if and only if $P_\ell$, $P_u$, and $P_L$ satisfy Assumption~\ref{assum:feas}.
\end{proposition}
\begin{proof}
   Under Assumption~\ref{assum:feas}, there exists $P_f \in \mathbb{R}^n$ such that $P_\ell < P_f < P_u$ and $\mathbbl{1}_n^\mathsf{T} P_f = \mathbbl{1}_n^\mathsf{T} P_L$. Next, we note that there exists $\theta \in \mathbb{R}^n$ such that $P_f-P_L = L \theta$ if and only if $(P_f-P_L) \perp \mathbbl{1}_n$ \cite[Lem.~6.12]{LNS} or, equivalently, if and only if $\mathbbl{1}_n^\mathsf{T} (P_f-P_L) = 0$ and sufficiency of Assumption~\ref{assum:feas} immediately follows. Next, note that there only exists $\theta \in \mathbb{R}^n$ such that $P_{\ell} < L\theta + P_{L} < P_{u}$ if   
   $\mathbbl{1}_n^\mathsf{T} P_\ell < \mathbbl{1}_n^\mathsf{T} (L \theta + P_L) < \mathbbl{1}_n^\mathsf{T} P_u$. Using $\mathbbl{1}_n^\mathsf{T} L = 0$, it directly follows that $\sum_{i=1}^n P_{\ell,i} < \sum_{i=1}^n  P_{L,i} < \sum_{i=1}^n P_{u,i}$ is necessary.
\end{proof}

Next, we substitute \eqref{eq:deCFP:dcpf} into \eqref{eq:deCFP:dcpf:obj} and \eqref{eq:deCFP:dcpf:lim} and scale \eqref{eq:deCFP:dcpf:lim}  by the diagonal matrix $K_I = \diag\{\sqrt{k_{I,i}}\}_{i=1}^n \in \mathbb{S}^n_{\succ 0}$, expanding the cost function, and dropping constant terms that do not depend on $\theta \in \mathbb{R}^n$, it can be shown that the set of optimizers $\theta \in \mathbb{R}^n$ of \eqref{eq:deCFP} is equivalent to the set of optimizers of
\begin{subequations}\label{eq:pfprob2}
    \begin{align}
        &\min_{\theta} \quad \tfrac{1}{2} \norm{L\theta}^2_{M} +  (P_L-P^\star)^\mathsf{T} M L \theta  \\
        & \text {s.t. } \quad  K_I P_\ell \leq K_I (L \theta + P_L) \leq K_I P_u. \label{eq:pfprob2:lim}
        \end{align}
\end{subequations}
We introduce the dual multipliers  $\lambda_{\ell} \coloneqq (\lambda_{\ell,1},\ldots,\lambda_{\ell,n}) \in \mathbb{R}^n_{\geq 0}$, $\lambda_{u} \coloneqq (\lambda_{u,1},\ldots,\lambda_{u,n})  \in \mathbb{R}_{\geq 0}^n$ associated with the constraints \eqref{eq:pfprob2:lim} of every node. Next, we define the set of points that satisfy the Karush-Kuhn-Tucker (KKT)~\cite[Ch.~5]{boyd2004convex} conditions of \eqref{eq:pfprob2}. Note that the stationary condition requires
\begin{align*}
    L \big(M (L \theta^\star + P_L-P^\star) + K_I (\lambda^\star_u-\lambda^\star_\ell)\big) = \mathbbl{0}_n.
\end{align*}
Then, using $\ker(L)=\ker(B^\mathsf{T})$, we can express the KKT conditions as follows.
\begin{definition}[\textbf{KKT points in nodal coordinates}]\label{def:stheta}
$\mathcal{S}_{\theta} \subseteq \mathbb{R}^{3n}$ denotes the points $(\theta^\star,\lambda^\star_\ell,\lambda^\star_u)$ that satisfy the KKT conditions of \eqref{eq:pfprob2}, i.e., $P_\ell \leq L \theta^\star + P_L \leq P_u$, $(\lambda^\star_\ell,\lambda^\star_u) \in \mathbb{R}^{2n}_{\geq 0}$, and 
\begin{subequations}\label{eq:KKT:nodal}
\begin{align}
    M (L \theta^\star + P_L-P^\star) + K_I (\lambda^\star_u-\lambda^\star_\ell) &\in \ker{B^\mathsf{T}},\label{eq:KKT:nodal:optimality}\\
    \diag\{\lambda^\star_{\ell,i}\}_{i=1}^n K_I (P_\ell - L \theta^\star - P_L)&=\mathbbl{0}_n,\\
    \diag\{\lambda^\star_{u,i}\}_{i=1}^n K_I (L \theta^\star + P_L - P_u)&=\mathbbl{0}_n.
\end{align}
\end{subequations}
\end{definition}
The next property directly follows from $\ker L = \mathbbl{1}_n$ and states that KKT points of \eqref{eq:pfprob2} are neither unique nor isolated.
\begin{property}[\textbf{Non-unique KKT points}]\label{property:stheta}
    For any $\left(\theta^\star, \lambda^\star\right) \in \mathcal{S}_{\theta}$ and all $c \in \mathbb{R}$ it holds that $\left(\theta^\star + cI_n, \lambda^\star\right) \in \mathcal{S}_{\theta}$.
\end{property}
To characterize the active constraints for KKT points of \eqref{eq:pfprob2} we require the following definition.
\begin{definition}[\textbf{Active constraint sets}] \label{def:activesets}
    We define $\mathcal{I}_{\ell} \subseteq \mc N$ and $\mathcal{I}_{u} \subseteq \mc N \setminus \mathcal{I}_{\ell}$ as the set of nodes at their lower and upper limit, i.e., $i \in \mathcal{I}_{\ell}$ if and only if $P_i = P_{\ell, i}$ and $i \in \mathcal{I}_{u}$ if and only if $P_i = P_{u, i}$.    
\end{definition}
The next result shows that, if the optimizer for any node $i \in \mathcal{N}$ is at the upper limit, then the optimizer for any node $j \in \mathcal{N} \setminus \{i\}$ cannot be at the lower limit and vice versa.
\begin{proposition}[\textbf{Mutually exclusive active sets}]\label{prop:mut:activesets}
    Consider $P_\ell$, $P_u$, $P_L$, and $P^\star$ such that Assumption~\ref{assum:feas} and Assumption~\ref{assum:setpoint} hold.  Then, for all $(\theta,\lambda) \in \mathcal{S}_{\theta}$, either $\mathcal{I}_{\ell} = \emptyset$ or $\mathcal{I}_{u} = \emptyset$.
\end{proposition}
\begin{proof}
    We will prove the result by contradiction. For any KKT point $(\eta^\star, \lambda^\star) \in \mc S_\eta$, it holds that
\begin{align*}
    M\left(BV\eta^\star + P_{L} - P^\star\right) + K_I (\lambda^\star_{u} - \lambda^\star_{\ell}) \in \ker(B^\mathsf{T}).
\end{align*}
Moreover, $\ker(B^\mathsf{T})\!=\!\vspan(\mathbbl{1}_n)$ and, by complementary slackness, $\lambda^\star_{\ell,i}=0$ and $\lambda^\star_{u,j}=0$ for any $(i,j) \in \mathcal{I}_{u} \times \mathcal{I}_{\ell}$. Thus,
\begin{align*}
    m_i\left(P_i -P^\star_{i}\right) &+ \sqrt{k_i}\lambda^\star_{u, i} = m_j\left(P_j -P^\star_{j}\right) - \sqrt{k_j}\lambda^\star_{\ell, j}
\end{align*} 
has to hold for all $(i,j) \in \mathcal{I}_{u} \times \mathcal{I}_{\ell}$. By feasibility of $(\eta^\star, \lambda^\star) \in \mathcal{S}_\eta$, we have $P_{i} = P_{u, i}$ for all $i \in \mathcal{I}_{u}$ and $P_{j} = P_{\ell, j}$ for all $j \in \mathcal{I}_{\ell}$. Then, dual-feasibility and Assumption~\ref{assum:setpoint} imply that
\begin{align*}
m_i\underbrace{\left(P_{u, i}- P_{i}^\star\right)}_{> 0}   + \underbrace{\sqrt{k_i}\lambda^\star_{u, i}}_{\geq 0} = m_i \underbrace{\left(P_{\ell, i} - P_{j}^\star\right)}_{< 0} - \underbrace{\sqrt{k_j}\lambda^\star_{\ell, j}}_{\geq 0},
\end{align*}
i.e., no pair $(i,j) \in \mathcal{I}_u \times \mc I_\ell$ can exist if $(\eta^\star, \lambda^\star) \in \mc S_\eta$.
\end{proof}
Moreover, we note that the active constraint set for any optimizer $(\theta,\lambda) \in \mathcal{S}_{\theta}$ can be further characterized in terms of the disturbances $P_L$ and references $P^\star$ as shown in Sec.~\ref{subsec:dynconstr}.
    
\subsection{Review of projections and primal-dual dynamics}
A common approach to solve \eqref{eq:pfprob2} in a distributed fashion is to leverage primal-dual dynamics~\cite{ashish} associated with \eqref{eq:pfprob2}. However, as illustrated below, applying standard primal-dual dynamics to \eqref{eq:pfprob2} results in an algorithm that requires the exchange of dual multipliers between nodes. We require the following definitions to formulate dynamics for solving \eqref{eq:dcpfeq}.
    \begin{definition}[\textbf{Normal and tangent cone}]\label{def:tangentcone}
    Given a non-empty convex set $\mathcal{C} \subseteq \mathbb{R}^n$, and a point $x \in \mathcal{C}$, the normal cone $\mathcal N_{x} \mathcal{C}$ is given by
    \begin{align*}
        \mathcal{N}_{x}\mathcal{C} \coloneqq \left\{w\in \mathbb{R}^n \mid w^\mathsf{T}\left(x^{\prime}-x\right) \leq 0, \quad \forall x^{\prime} \in \mathcal{C}\right\}.
    \end{align*} 
    Then, the tangent cone of the set $\mathcal{C}$ at the point $x$ is defined as the polar cone of the normal cone
    \begin{align*}
        \mathcal{T}_{x}\mathcal{C} \coloneqq \left\{v \in \mathbb{R}^n \mid v^\mathsf{T} w \leq 0,\quad \forall w \in \mc N_x \mc C\right\}.
    \end{align*} 
\end{definition}
Next, we define the projection operator.
\begin{definition}[\textbf{Projection}]\label{def:projection}
    Given a convex set $\mathcal{C} \subseteq \mathbb{R}^n$ and a vector $v \in \mathbb{R}^n$, $\Pi_{\mathcal{C}}(v)$ denotes the projection of $v$ with respect to the set $\mathcal{C}$, i.e., $\Pi_{\mathcal{C}}(v) = \argmin\nolimits_{p \in \mathcal{C}} \norm{p - v}$.
\end{definition}
Broadly speaking, projecting a dynamical system $\ddt x = f(x)$ onto a  set $\mc C$ results in the projected dynamical system $\ddt x = \Pi_{\mathcal{T}_{x}\mathcal{C}}(f(x))$ that does not leave the set $\mc C$~\cite{DN1993}. While inherently discontinuous, strong theoretical results on the stability and convergence properties of projected dynamical systems are available in the literature (see, e.g., \cite{HBD2021}).

Consider the CFP \eqref{eq:deCFP}, a regularization parameter $\rho \in \mathbb{R}_{>0}$, and the augmented Lagrangian 
$\mathcal{L}_\rho \coloneqq \tfrac{1}{2} \norm{P-P^\star}_M + \lambda_{u}^\mathsf{T}(P-P_u) + \lambda_{\ell}^\mathsf{T} (P_\ell -P) + \frac{\rho}{2} (\|\Pi_{\mathbb{R}^n_{\geq 0}}\left(P_{\ell} - P\right)\|^2 + \|\Pi_{\mathbb{R}^n_{\geq 0}} \left(P - P_{u} \right)\|^2)$.

Applying primal-dual gradient dynamics \cite{ashish} to $\mathcal{L}_\rho$ results in the \emph{distributed} dynamics
\begin{subequations}
    \begin{align}
        \ddt \theta =& -L M (P -P^\star) - \rho L  \Pi_{\mathbb{R}^n_{\geq 0}}\left(P - P_{u} \right) - L  \lambda_{u}
        \label{eq:pdflow:primal}\\
         & + \rho L  \Pi_{\mathbb{R}^n_{\geq 0}} \left(P_{\ell} - P \right)  + L  {\lambda_\ell}, \nonumber\\
        \ddt \lambda_{\ell,i} =& \Pi_{\mathcal{T}_{\lambda_{\ell,i}} \mathbb{R}_{\geq 0}}(P_{\ell,i} - P_i), \label{eq:pdflow:lambda:l}\\
        \ddt \lambda_{u,i} =& \Pi_{\mathcal{T}_{\lambda_{u,i}} \mathbb{R}_{\geq 0}}(P_i - P_{u,i}).\label{eq:pdflow:lambda:u}
    \end{align}
\end{subequations}
Notably, the dynamics of the dual multipliers \eqref{eq:pdflow:lambda:l} and \eqref{eq:pdflow:lambda:u} can be implemented at every node $i \in \mathcal{N}$ using only local information. However, the primal dynamics \eqref{eq:pdflow:primal} cannot be implemented using only local information, e.g., evaluating $L \lambda_{\ell}$ requires exchanging dual multipliers between nodes.

\subsection{Networked dynamics}
In the remainder of this paper, we will focus on the \emph{networked} dynamics 
\begin{subequations}\label{eq:plimdroopsc}
    \begin{align}
    \!\!\!\!\!\!\!\!\!\ddt \theta_i =& m_i (P^\star_i-P_i)\! -\! k_{P,i} \Pi_{{\mathbb{R}}_{\geq 0}}(P_i-P_{u,i}) \label{eq:plimdroop:angle}\\ &+ k_{P,i} \Pi_{{\mathbb{R}}_{\geq 0} }(P_{\ell,i} - P_i) - \sqrt{k_{I,i}} (\lambda_{u,i}-\lambda_{\ell,i}), \nonumber\\
    \!\!\!\!\!\!\!\!\!\ddt\lambda_{\ell,i} =&  \Pi_{\mathcal{T}_{\lambda_{\ell,i}} \mathbb{R}_{\geq 0}} \left(\sqrt{k_{I,i}}(P_{\ell,i} - P_i)\right),\\
    \!\!\!\!\!\!\!\!\!\ddt\lambda_{u,i} =&  \Pi_{\mathcal{T}_{\lambda_{u,i}} \mathbb{R}_{\geq 0}} \left(\sqrt{k_{I,i}}(P_i - P_{u,i})\right).
    \end{align}
\end{subequations}
for all $i \in \mathcal{N}$ with gains $k_{P,i} \in \mathbb{R}_{>0}$ and $k_{I,i} \in \mathbb{R}_{>0}$ and \emph{local} controller states $\theta_i$,  $\mu_{u,i} \in \mathbb{R}_{\geq 0}$, and $\mu_{\ell,i} \in \mathbb{R}_{\geq 0}$. The main contribution of this work is to show that, under mild conditions, \eqref{eq:plimdroopprime} converges to the set of KKT points of \eqref{eq:pfprob2}.

\subsection{Main results}
Before stating our main result, we require the following definition of a forward invariant set.
\begin{definition}[\textbf{Forward invariant set}]\label{def:finv}
A set $\mc D$ is called forward invariant under the dynamics $\ddt x = f(x)$ if, for any $x_0 \in \mc D$, it holds that $\varphi_x(t,x_0) \in\mc D$ for all $t \in \mathbb{R}_{>0}$.
\end{definition}
In other words, Definition~\ref{def:finv} requires that the solution $\varphi_x(t,x_0)$ of a dynamical system remains in $\mc D$ for all times if the initial condition $x_0$ is in  $\mc D$. Moreover, we require the following definition of stability with respect to a set.

\begin{definition}[\textbf{Asymptotic stability with respect to a set}]\label{def:GAS}
Given a dynamic system $\ddt x = f(x)$ and forward invariant set $\mc D$, $\ddt x = f(x)$ is called globally asymptotically stable with respect to a set $\mc C \subseteq \mc D$ in $\mc D$ if 
\begin{enumerate}[label=(\roman*)]
 \item it is globally attractive with respect to $\mc C$, i.e., $
     \lim_{t\to\infty} \norm{\varphi_{\blue x}(t,x_0)}_{\mc C} = 0$ holds for all $x_0 \in \mc D$, and \label{def:GAS:attr}
        \item it is Lyapunov stable with respect to $\mc C$, i.e., for every $\varepsilon \in \mathbb{R}_{>0}$ there exists $\delta \in \mathbb{R}_{>0}$ such that $x_0 \in \mc D$ and $\norm{x_0}_{\mc C} < \delta$ implies $\norm{\varphi_{x}(t,x_0)}_{\mc C} < \varepsilon$ for all $t \in \mathbb{R}_{\geq 0}$.\label{def:GAS:stab}
\end{enumerate}
\end{definition}
Notably, the definition of asymptotic stability typically assumes compactness of the set $\mc C$~\cite{A2004}. Because Definition~\ref{def:GAS} does not require $\mc C$ to be compact, stability with respect to $\mc C$ does not necessarily imply convergence to a limit cycle or an equilibrium, but trajectories may tend to infinity within the set $\mc C$. In the application at hand, the dynamics not bounded by $\mc C$ correspond to the synchronous frequency dynamics upon convergence, which we will study separately.

Using standard arguments (see~\cite[Ch.~V]{H67,K2014}) one can show that Definition~\ref{def:GAS} is identical to the following condition.
\begin{definition}{\bf(Comparison functions)}
    A function $\chi:\mathbb{R}_{\geq0} \to \mathbb{R}_{\geq0}$ is of class $\mathscr{K}$ if it is continuous, 
   strictly increasing and $\chi(0)=0$. A function $\chi:  \mathbb{R}_{\geq0} \to \mathbb{R}_{>0}$ is of class $\mathscr{L}$ if it is continuous, non-increasing, and $\chi(s) \to 0$ as $s \to \infty$.
\end{definition}
\begin{condition}[\textbf{Comparison function characterization}]\label{def:compGAS}
    Consider a function $\chi(\norm{x_0}_{\mc C},t) \in \mathscr{KL}$, i.e., it is of class $\mathscr{K}$ in its first argument and class $\mathscr{L}$ in its second argument. Asymptotic stability with respect to $\mc C$ in $\mc D$ is equivalent to 
    \begin{align*}
    \norm{\varphi_x(t,x_0)}_{\mc C} \leq \chi(\norm{x_0}_{\mc C},t), \quad \forall x_0 \in \mc D, \forall t \in \mathbb{R}_{\geq 0}.
    \end{align*}
    \end{condition}

If $\mc D = \mathbb{R}^n$, then Definition~\ref{def:GAS} implies global asymptotic stability. For the system \eqref{eq:plimdroopsc}, $\mc D = \mathbb{R}^n \times \mathbb{R}^{2n}_{\geq 0} \neq \mathbb{R}^n$ but contains all initial conditions with non-negative dual multipliers. Thus, with a slight abuse in terminology, we will refer to the system as globally asymptotically stable for $\mc D = \mathbb{R}^n \times \mathbb{R}^{2n}_{\geq 0}$. To characterize the dynamics on the set $\mc C = \mathcal{S}_{\theta}$, we introduce $\omega = \ddt \theta$. We are now ready to state our main stability result. 

\begin{theorem}[\textbf{Global asymptotic stability}]\label{thm:GASpowerlimit}
    Consider $P_\ell$, $P_u$, $P_L$, and $P^\star$ such that Assumption~\ref{assum:feas} and Assumption~\ref{assum:setpoint} hold. For any connected graph $\mc G$, \eqref{eq:plimdroopsc} is globally asymptotically stable on $\mathbb{R}^n \times \mathbb{R}^{2n}_{\geq 0}$ with respect to the set  $\mathcal{S}_{\theta}$. Moreover, there exists $\omega_s \in \mathbb{R}$ such that $\lim_{t \rightarrow \infty} \omega_i(t) = \omega_{s}$ for all $i \in \mc N$.
    %
\end{theorem}
A proof is provided in Section~\ref{sec:stabanalysis}. Theorem~\ref{thm:GASpowerlimit} shows that \eqref{eq:plimdroopsc} is Lyapunov stable with respect to the set $\mathcal{S}_{\theta}$ of KKT points of \eqref{eq:pfprob2} and converges to an optimizer in $\mathcal{S}_{\theta}$ as $t\to\infty$. Notably, by Property~\ref{property:stheta} this implies that \eqref{eq:plimdroopsc} is globally asymptotically stable with respect to a synchronous motion in the set $\mathcal{S}_{\theta}$ but not necessarily with respect to an equilibrium point.

\begin{remark}[\textbf{Convergence rate}]
    Because our results leverage the LaSalle function-based results from \cite{ashish}, they do not provide a convergence rate. We conjecture that replacing the projection of the dual multipliers in the networked dynamics \eqref{eq:plimdroopsc} with projection-free dual multiplier dynamics introduced in~\cite{QL2019} could provide a pathway to establish exponential convergence.
\end{remark}

\section{Stability analysis}\label{sec:stabanalysis}
In this section, we present the stabiltiy analysis and proofs that establish our main results stated in the previous section.

\subsection{Overview and proof strategy}
To establish the equivalence between optimizers of the CFP \eqref{eq:pfprob2} and limit points of the networked dynamics \eqref{eq:plimdroopsc} we will use the proof strategy shown in Fig.~\ref{fig:strategy}. Our results crucially depend on two main steps.

First, we use the oriented incidence matrix $B$ and decomposition $V\in\mathbb{R}^{e \times e}$ of the weight matrix $W=VV \in\mathbb{R}^{e \times e}$ of the graph $\mc G$ to define the change of coordinates $\eta = VB^\mathsf{T} \theta$. Notably, this change of coordinates transforms nodal angles to angle differences across graph edges (i.e., transmission lines). Applying this change of coordinates to the CFP \eqref{eq:pfprob2} results in an optimization problem in edge coordinates whose KKT points can be related to the KKT points of the CFP \eqref{eq:pfprob2} under the restriction $V \in \Ima(B^\mathsf{T})$. We emphasize that, beyond providing an interpretation of the networked dynamics \eqref{eq:plimdroopsc}, our results crucially hinge on the equivalence of the networked dynamics \eqref{eq:plimdroopsc} and primal-dual dynamics of the CFP \eqref{eq:pfprob2} after applying the change to edge coordinates.

Second, we show that, in edge coordinates $\eta$, the networked dynamics \eqref{eq:plimdroopsc} can be interpreted as a distributed primal-dual algorithm solving \eqref{eq:pfprob2} while maintaining $\eta(t) \in \Ima (VB^\mathsf{T})$ for all times $t \in \mathbb{R}_{\geq 0}$ if $\eta(0) \in \Ima (VB^\mathsf{T})$. Notably, the dynamics in edge coordinates can be decomposed into (i) dynamics associated with primal-dual dynamics resulting from a strictly convex problem, and (ii) remaining dynamics that are stable in the sense of Lyapunov. The primal-dual dynamics resulting from a strictly convex problem can then be analyzed using well-known results from~\cite{ashish}. Notably, this second change of coordinates is required to decompose the system into its  asymptotically stable dynamics and null dynamics that, in the context of power systems, correspond to invariance of the dynamics under rotation.

Combining the aforementioned results allows us to establish that the networked dynamics \eqref{eq:plimdroopsc} are globally asymptotically stable with respect to the set of KKT points of the constrained flow problem \eqref{eq:pfprob2}. 
\begin{figure*}[htbp]
    \begin{center}
            \vspace{-1em}
        \includegraphics[width=0.9\linewidth]{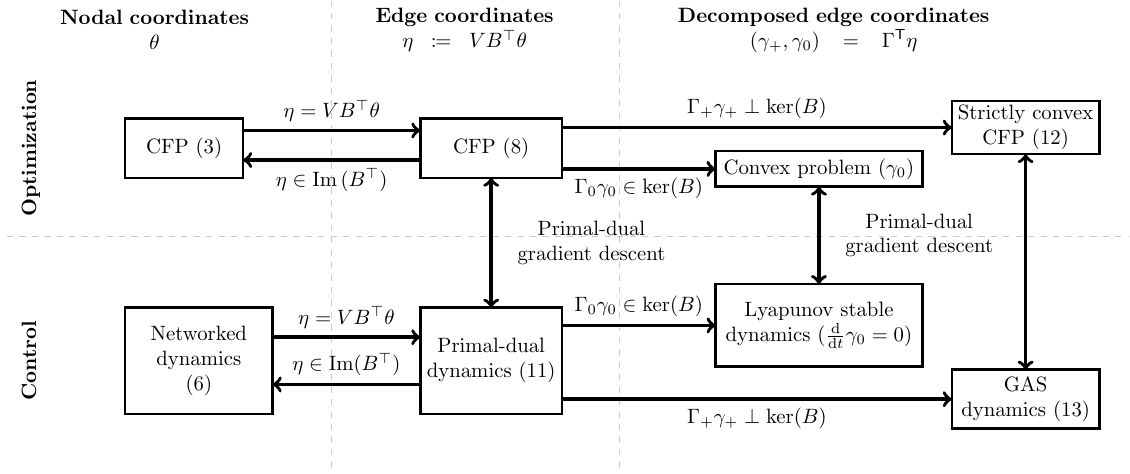}
        \caption{The networked dynamics in edge coordinates coincide with the primal-dual dynamics associated of the CFP in edge coordinates. To establish our main result, we show that the dynamics in edge and node coordinates are identical up to the dynamics of $\eta \in \ker(B^\mathsf{T})$. \label{fig:strategy}}
    \end{center}
\end{figure*}

\subsection{Constrained flow problem in edge coordinates}
To establish our main result, we reformulate the constrained flow problem \eqref{eq:pfprob2} in edge coordinates. To this end, consider the (weighted) differences  $\eta \coloneqq VB^\mathsf{T} \theta \in \mathbb{R}^{e}$ between nodal variables and the decomposition
\begin{align}\label{eq:coordination_changing}
    L\theta = BV V B^\mathsf{T} \theta = BV \eta
\end{align}
of the Laplacian matrix $L$ into its oriented incidence matrix $B \in \mathbb{R}^{n \times e}$ and weight matrix $V \coloneqq W^{\frac{1}{2}} \in \mathbb{R}^{e \times e}$. Applying  \eqref{eq:coordination_changing} to \eqref{eq:pfprob2} results in the constrained flow problem in \emph{edge coordinates} 
\begin{subequations}\label{eq:pfangledifference}
\begin{align}
    &\min_\eta \tfrac{1}{2} \norm{ BV\eta}^2_{M}   +\left(P_L-P^\star\right)^{\mathsf{T}} M BV\eta   \\ 
    & \text{s.t.} \quad K_I P_{\ell} \leq K_I (BV\eta+P_L)  \leq K_I P_u.\end{align}
\end{subequations}
Notably, the Hessian $V B^\mathsf{T} M BV \in \mathbb{S}^{e}_{\succeq 0}$ of \eqref{eq:pfangledifference} becomes a weighted edge Laplacian matrix (see~\cite{ZM2011,ZB2014} for details) if $M=c I_e$ for some $c \in \mathbb{R}_{>0}$. We will show that $V B^\mathsf{T} M BV \in \mathbb{S}^{e}_{\succ 0}$ (i.e., \eqref{eq:pfangledifference} is strictly convex) when $\mc G$ has no cycles but $V B^\mathsf{T} M BV \in \mathbb{S}^{e}_{\succeq 0}$ otherwise. In other words, $\eta = VB^\mathsf{T} \theta$ is generally not a similarity transform. Before investigating this aspect further, the same steps as in the proof of Proposition~\ref{prop:feas} can be used to show that \eqref{eq:pfangledifference} admits a feasible solution under Assumption~\ref{assum:feas}.
\begin{proposition}[\textbf{Feasibility in edge coordinates}]\label{prop:feasdiff}
    There exists $\eta \in \mathbb{R}^n$ such that $P_{\ell} < B V \eta + P_{L} < P_{u}$ if and only if $P_\ell$, $P_u$, and $P_L$ satisfy Assumption~\ref{assum:feas}.
\end{proposition}
Next, we characterize the optimizers of \eqref{eq:pfangledifference}. Note that the stationary condition requires
\begin{align*}
    B^\mathsf{T}\big(M (BV \eta^\star + P_L - P^\star) + K_I (\lambda^\star_u-\lambda^\star_\ell)\big) = \mathbbl{0}_{e}.
\end{align*}
Thus, the KKT conditions can be expressed as follows.
\begin{definition}[\textbf{KKT points of CFP in edge coordinates}]\label{def:seta}
$\mathcal{S}_{\eta} \subseteq \mathbb{R}^{e + 2n}$ denotes the set of points $(\eta^\star,\lambda^\star_\ell,\lambda^\star_u)$ that satisfy the KKT conditions of the CFP in edge coordinates \eqref{eq:pfangledifference}, i.e., $P_\ell \leq BV \eta^\star + P_L \leq P_u$,  $(\lambda^\star_\ell,\lambda^\star_u) \in \mathbb{R}^{2n}_{\geq 0}$, and 
\begin{subequations}\label{eq:KKT:edge} 
\begin{align} 
    M (BV \eta^\star + P_L - P^\star) + K_I (\lambda^\star_u-\lambda^\star_\ell) &\in \ker(B^\mathsf{T}),\label{eq:KKT:edge:optimality}\\
    \diag\{\lambda^\star_{u,i}\}_{i=1}^n K_I (BV \eta^\star + P_L - P_u)&=\mathbbl{0}_n,\\
    \diag\{\lambda^\star_{\ell,i}\}_{i=1}^n K_I (P_\ell - BV \eta^\star - P_L)&=\mathbbl{0}_n.
\end{align}
\end{subequations}
\end{definition}
The following result clarifies the relationship between KKT points of the constrained flow problem \eqref{eq:pfprob2} in nodal coordinates and the constrained flow problem \eqref{eq:pfangledifference} in edge coordinates.
\begin{proposition}[\textbf{KKT points in edge coordinates}]\label{prop:equiv}\phantom{a}
\begin{enumerate}[label=(\roman*)]
    \item For any $(\eta^\star, \lambda^\star) \in \mathcal{S}_{\eta}$, there exist $\eta^\star=VB^\mathsf{T} \theta^\star$ such that $(\theta^\star, \lambda^\star) \in \mathcal{S}_\theta$ if and only if $\eta^\star \in \Ima(B^\mathsf{T})$.
        \item $(\theta^\star, \lambda^\star) \in \mathcal{S}_\theta$ if and only if $(VB^\mathsf{T} \theta^\star, \lambda^\star) \in \mathcal{S}_{\eta}$.
\end{enumerate}
\end{proposition}
\begin{proof}
Note that $\theta^\star \in \mathbb{R}^n$ such that $\eta^\star = V B^\mathsf{T}\theta^\star$ exists if and only if $\eta^\star \in \Ima(B^\mathsf{T})$. Then, the first statement immediately follows by substituting $\theta^\star = V B^\mathsf{T}\eta^\star$ into the equations defining $\mathcal{S}_{\theta}$ and noting that $\ker(L)=\ker(B^\mathsf{T})$. To show the second statement, substitute $(\eta^\star, \lambda^\star) = (VB^\mathsf{T}\theta^\star, \lambda^\star)$ into \eqref{eq:KKT:edge}. Then, both $(\theta^\star, \lambda^\star) \in \mathcal{S}_\theta$ and $(VB^\mathsf{T} \theta^\star, \lambda^\star) \in \mathcal{S}_{\eta}$ hold if and only if $M(L \theta^\star \!+\! P_L \!+\!P^\star)\!+\! K_I (\lambda^\star_u\!-\!\lambda^\star_\ell) \in \ker(B^\mathsf{T})$.
\end{proof}
In other words, $\eta^\star \in \Ima (VB^\mathsf{T})$ ensures that the angle differences $\eta^\star \in \mathbb{R}^e$ are restricted to the set $\Ima(B^\mathsf{T})$ for which a corresponding angle configuration $\theta^\star \in \mathbb{R}^n$ can be found. Moreover, the sets of KKT points of \eqref{eq:pfprob2} and \eqref{eq:pfangledifference} coincide under the edge transformation $\eta = VB^\mathsf{T} \theta$.

\subsection{Networked dynamics in edge coordinates}
For simplicity of the notation, we first vectorize the networked dynamics \eqref{eq:plimdroopsc} in nodal coordinates to obtain
\begin{subequations}\label{eq:plimdroop}
    \begin{align}
        \ddt\theta =& M\left(P^\star-P_L-L \theta\right) - \left(\alpha \otimes K_I \right)\lambda \label{eq:plimdroop:theta}\\ 
            &-\left( \alpha \otimes K_P \right) \Pi_{{\mathbb{R}}^{2n}_{\geq 0} }\left(g(L \theta)\right), \nonumber \\
            \ddt\lambda =&  \Pi_{\mathcal{T}_{\lambda} \mathbb{R}^{2n}_{\geq 0}} \big(\left(I_2 \otimes K_{I}\right) g(L\theta)\big), \label{eq:plimdroop:lambda}
    \end{align}
\end{subequations}
where $\alpha \coloneqq (-1, 1)^\mathsf{T}$, $\lambda \coloneqq  (\lambda_\ell,\lambda_u) \in \mathbb{R}^{2n}_{\geq 0}$, and the function $g: \mathbb{R}^n \to \mathbb{R}^{2n}$ and network power injection $P_N \coloneqq L \theta$ are used to express the violation of the inequality constraints as
\begin{align*}
    g(P_N) \coloneqq \begin{bmatrix}
    P_{\ell} - P_N - P_L\\
    P_N + P_{L}-P_u
    \end{bmatrix}.
\end{align*}
Next, we will investigate stability of primal-dual gradient descent applied to the CFP in edge coordinates \eqref{eq:pfangledifference}. The augmented Lagrangian associated with \eqref{eq:pfangledifference} is given by
\begin{align*}
   \mathcal{L}(\eta, \lambda)\coloneqq& \tfrac{1}{2} \norm{BV \eta}^2_{M}+\left(P_L-P^\star\right)^{\mathsf{T}} M BV \eta \\
    +&\tfrac{1}{2} \norm{\Pi_{\mathbb{R}_{\geq 0}^{2n}} \left(g(BV\eta)\right)}^2_{I_2 \otimes K_P}  + \lambda^\mathsf{T} K_I  g(BV\eta),
    \end{align*}
Next, we introduce the primal-dual gradient dynamics $\ddt\eta = -\nabla_{\eta} \mathcal{L}$, $\ddt\lambda = \Pi_{\mathcal{T}_{\lambda} \mathbb{R}^{2n}_{\geq 0}} \left(\nabla_{\lambda}\mathcal{L}\right)$ associated with the augmented Lagrangian. This results in  
\begin{subequations}\label{eq:primaldualedge}
    \begin{align}
        \ddt\eta =&   VB^\mathsf{T} (M(P^\star-P_L - BV\eta) -  (\alpha \otimes K_I)\lambda,  \label{eq:primaldualedge:angle} \\
        & - \left(\alpha \otimes K_P\right) \Pi_{{\mathbb{R}}^{2n}_{\geq 0} }\left(g(BV\eta)\right)), \nonumber \\
        \ddt\lambda =&    \Pi_{\mathcal{T}_{\lambda} \mathbb{R}^{2n}_{\geq 0}} \left((I_2 \otimes K_I) g(BV\eta)\right).
    \end{align}
\end{subequations}
The next theorem shows the primal-dual dynamics \eqref{eq:primaldualedge} are globally asymptotically stable with respect the set of KKT points $\mc S_{\eta}$. In addition, we show that the primal-dual dynamics \eqref{eq:primaldualedge} converge to an equilibrium.

\begin{theorem}[\textbf{Global asymptotic stability of primal-dual dynamics in edge coordinates}]\label{thm:edgeconvergence}
    Consider $P_\ell$, $P_u$, $P_L$, and $P^\star$ such that Assumption~\ref{assum:feas} and Assumption~\ref{assum:setpoint} hold. Then the primal-dual dynamics \eqref{eq:primaldualedge} are globally asymptotically stable with respect to $\mc S_\eta$ on $\mathbb{R}^{e} \times \mathbb{R}^{2n}_{\geq0}$. Moreover, $\ddt (\eta,\lambda) = \mathbbl{0}_{e+2n}$ holds on $\mc S_\eta$. 
\end{theorem}
\begin{proof} We begin by noting that $M \in \mathbb{S}^n_{\succ 0}$. Then, by \cite[Observation~7.1.8]{HJ2013}, $B^\mathsf{T} M B \in \mathbb{S}^{e \times e}_{\succ 0}$ if and only if $\rank{B}=e$. If $\mc G$ is a connected tree, then $n=e+1$ and by \cite[Lemma 9.2]{LNS}, $\rank{B}=e$. Conversely, if $\mc G$ contains cycles, then $e \geq n$ and $\rank{B} \leq e-1$. Thus, if $\mc G$ is a tree, then the cost function of \eqref{eq:pfangledifference} is strictly convex and $\mc S_\eta$ is a singleton. Moreover, by Proposition~\ref{prop:feasdiff} there exists $\eta$ such that $P_\ell < B V \eta + P_L < P_u$, i.e., Slater's condition holds. Then, \cite[Theorem~4.5]{ashish} immediately implies that \eqref{eq:primaldualedge} is globally asymptotically stable with respect to $\mc S_\eta$. 
        
When $\mc G$ contains cycles, we can decompose \eqref{eq:pfangledifference} and \eqref{eq:primaldualedge} into a strictly convex part and remaining dynamics. To this end, let $\Gamma \coloneqq \begin{bmatrix} \Gamma_{+} & \Gamma_{0} \end{bmatrix}$ where $\Gamma_{+}  \in \mathbb{R}^{e \times n-1}$ contains eigenvectors corresponding to the positive eigenvalues of $V B^\mathsf{T} M B V$ and $\Gamma_{0}  \in \mathbb{R}^{e \times e-(n-1)}$ contains the eigenvectors corresponding to the zero eigenvalues. Next, let $\gamma = (\gamma_+,\gamma_0) \in \mathbb{R}^{e}$. Since $B^\mathsf{T} M B \in \mathbb{S}^n_{\succeq 0}$, we conclude that $\Gamma^{-1} =  \Gamma^\mathsf{T}$. Applying the change of coordinates $\eta = \Gamma \gamma$ to \eqref{eq:pfangledifference} results in
\begin{subequations}\label{eq:pfangledifferenceplus}
    \begin{align}
        &\min_\eta \tfrac{1}{2} \norm{\gamma_+}^2_{H} + c^\mathsf{T} \gamma_+ \\ 
        & \text{s.t. }   K_I P_{\ell} \leq K_I (A \gamma_+ + P_L)  \leq K_I  P_u,
    \end{align}
\end{subequations}
where $H\coloneqq\Gamma_+^\mathsf{T} VB^\mathsf{T}MBV \Gamma_+$, $c\coloneqq\Gamma_+^\mathsf{T}VBM(P^\star\!-\!P_L)$,  and $A\coloneqq BV \Gamma_+$. Notably, this transformation only removed redundant degrees of freedom and, by construction, \eqref{eq:pfangledifferenceplus} is strictly convex and strictly feasible under the same conditions as \eqref{eq:pfangledifference}. Moreover, given a KKT point $(\gamma^\star_+,\lambda^\star)$ of \eqref{eq:pfangledifferenceplus}, $BV \Gamma_0 \in \mathbb{R}^{n \times e-(n-1)}$ implies that $(\Gamma_+ \gamma^\star_+ + \Gamma_0 \gamma_0,\lambda^\star) \in \mc S_\eta$ for all $\gamma_0 \in \mathbb{R}^{e-(n-1)}$. Applying the change of coordinates $\eta = \Gamma \gamma$ to \eqref{eq:primaldualedge} results in $\ddt \gamma_0 =0$ and 
\begin{subequations}\label{eq:primaldualedgeplus}
\begin{align}
    \!\!\ddt \gamma_+=&-\!H \gamma_+ \!-\!c \!-\! \Gamma_+^\mathsf{T} \big( (\alpha \otimes K_{P}) \Pi_{\mathbb{R}^{2n}_{\geq 0}}(g(A \gamma_+)) \nonumber\\
    & +  (\alpha \otimes K_I) \lambda \big),  \\ 
    \ddt \lambda =& \Pi_{\mathcal{T}_{\lambda}{\mathbb{R}^{2n}_{\geq 0} }}\left(K_I g(A\gamma_+)\right).
\end{align}
\end{subequations}
Notably, \eqref{eq:primaldualedgeplus} corresponds to primal-dual dynamics of the augmented Lagrangian of \eqref{eq:pfangledifferenceplus}. Thus, by \cite[Theorem~4.5]{ashish}, the dynamics \eqref{eq:primaldualedgeplus} are globally asymptotically stable with respect to a KKT point $(\gamma^\star_+,\lambda^\star)$ of \eqref{eq:pfangledifferenceplus}. In other words, \eqref{eq:primaldualedge} can be decomposed into dynamics that are globally asymptotically stable with respect to $(\gamma^\star_+,\lambda^\star)$ and a constant $\gamma_0$. Since $(\eta,\lambda) = (\Gamma_+ \gamma_+ + \Gamma_0 \gamma_0,\lambda) \in \mc S_\eta$ for any $\gamma_0 \in \mathbb{R}^{e \times e-(n-1)}$, it follows that \eqref{eq:primaldualedge} is globally asymptotically stable with respect to $\mc S_\eta$. The last statement of the Theorem follows by noting that $\ddt (\gamma_+,\lambda)=\mathbbl{0}_{3n-1}$ when $(\gamma_+,\lambda)=(\gamma^\star_+,\lambda^\star)$ and $\ddt \gamma_0=0$.
\end{proof}
The following corollary is a direct consequence of the proof of Theorem~\ref{thm:edgeconvergence} and establishes that the optimizer of \eqref{eq:pfangledifference} is unique if the graph $\mc G$ is a tree.
\begin{corollary}[\textbf{Radial network}]\label{corr:radial}
    Consider $P_\ell$, $P_u$, $P_L$, and $P^\star$ such that Assumption~\ref{assum:feas} and Assumption~\ref{assum:setpoint} hold. If $\mc G$ is a tree, then \eqref{eq:primaldualedge} is globally asymptotically stable with respect to the unique optimizer of \eqref{eq:pfangledifference}, i.e., $\mc S_\eta$ is a singleton.
\end{corollary}

To analyze the networked dynamics in edge coordinates, let $C_\eta \coloneqq \begin{bmatrix} I_e & \mathbbl{0}_{e \times 2n} \end{bmatrix}$ and $T_\eta \coloneqq \blkdiag(VB^\mathsf{T},I_{2n})$.

\begin{lemma}[\textbf{Coinciding vector fields}]\label{lem:coincide}
    Let $\varphi_{\theta}(t,(\theta_0,\lambda_0))$ and $\varphi_{\eta}(t,(\eta_0,\lambda_0))$ denote the solutions of \eqref{eq:plimdroop} and \eqref{eq:primaldualedge} for initial conditions $(\theta_0,\lambda_0)$ and $(\eta_0,\lambda_0)$. Then, it holds that
    \begin{enumerate}[label=(\roman*)]
        \item $\eta(t)=C_\eta \varphi_{\eta}(t,(\eta_0,\lambda_0)) \in \Ima(VB^\mathsf{T})$ for all $t \in \mathbb{R}_{\geq0}$ and all $\eta_0\in \Ima(VB^\mathsf{T})$, and \label{lem:coincide:invariant}
        \item $T_\eta \varphi_{\theta}(t,(\theta_0,\lambda_0))=\varphi_{\eta}(t,T_\eta(\theta_0,\lambda_0))$. \label{lem:coincide:solutions}
    \end{enumerate} 
\end{lemma}
\begin{proof}
We first note that $\ddt \eta \in \Ima(VB^\mathsf{T})$ in \eqref{eq:primaldualedge}. Therefore, for all $\eta_0 \in \Ima(VB^\mathsf{T})$, it holds that $C_\eta \varphi_\eta(t,(\eta_0,\lambda_0)) \in \Ima(VB^\mathsf{T})$ for all $t \in \mathbb{R}_{\geq 0}$. To show statement~\ref{lem:coincide:solutions}, let 
\begin{align*}
    f(\eta, \lambda) \coloneqq& M\left(P^\star-P_L-B V \eta\right)- (\alpha \otimes K_I) \lambda    \\
    & -(\alpha \otimes K_P)\Pi_{\mathbb{R}_{\geq 0}^n}\left(g(BV\eta)\right).
\end{align*}
Then \eqref{eq:plimdroop} and \eqref{eq:primaldualedge} can be written as
\begin{align}\label{eq:sigma_eta}
        VB^\mathsf{T} \ddt \theta = VB^\mathsf{T} f \big(\underbrace{VB^\mathsf{T}\theta}_{=\eta}, \lambda\big) = 
        \ddt \eta, \\
        \ddt \lambda = \Pi_{\mathcal{T}_{\lambda} \mathbb{R}^{2n}_{\geq 0}}\bigg((I_2 \otimes K_I) g\big(BV\underbrace{VB^\mathsf{T}\theta}_{=\eta}\big)\bigg).
\end{align}
In other words, the vector fields of \eqref{eq:plimdroop} and \eqref{eq:primaldualedge} coincide mapped to the edge coordinates in the sense of statement~\ref{lem:coincide:solutions} (i.e., by multiplying \eqref{eq:plimdroop} with $T_\eta$ from the left) when $\eta \in \Ima(VB^\mathsf{T})$.
\end{proof}
In other words, when starting from an initial condition such that $\eta_0 \in \Ima(VB^\mathsf{T})$, the dynamics \eqref{eq:primaldualedge} coincide with the dynamics \eqref{eq:plimdroop} mapped to the edge coordinates.

\subsection{Proof of the main results}
We are now ready to prove our main result. In particular, we leverage Theorem~\ref{thm:edgeconvergence} to establish stability of the overall networked dynamics with respect to the set of KKT points $\mc S_\eta$ of the CFP \eqref{eq:pfprob2}. Notably, by applying the two changes of coordinates shown in the top half of Fig.~\ref{fig:strategy}, the CFP \eqref{eq:pfprob2} can be decomposed into a strictly convex and convex problem. Theorem~\ref{thm:edgeconvergence} establishes global asymptotic stability of the dynamics in edge-coordinates associated with the strictly convex part of the CFP. The proof of Theorem~\ref{thm:GASpowerlimit} builds on this result and (i) accounts for the remaining Lyapunov stable dynamics in edge coordinates, and (ii) establishes stability in the original coordinates.

\textit{Proof of Theorem~\ref{thm:GASpowerlimit}:}
First, we establish that there exist $\underline{\kappa} \in \mathbb{R}_{\geq 0}$ and $\overline{\kappa} \in \mathbb{R}_{\geq 0}$, $\forall t\geq 0$
 such that  
\begin{align}\label{eq:normbound}
    \underline{\kappa} \norm{T_\eta \varphi_\theta(t,\xi_0)}_{\mc S_\eta} \!\leq\! \norm{\varphi_\theta(t,\xi_0)}_{\mc S_\theta}  \!\leq\! \overline{\kappa} \norm{T_\eta \varphi_\theta(t,\xi_0)}_{\mc S_\eta}
\end{align}
holds for all $\xi_0 \in \mathbb{R}^n \times \mathbb{R}^{2n}_{>0}$. To this end, we note that 
\begin{align}
    \norm{T_\eta (\theta,\lambda)}_{\mc S_\eta}&=\min\nolimits_{(\eta^\prime,\lambda^\prime)\in\mc S_\eta} \norm{(\eta^\prime -VB^\mathsf{T}\theta,\lambda^\prime-\lambda)},\label{eq:ps:eta}\\
    \norm{(\theta,\lambda)}_{\mc S_\theta}&=\min\nolimits_{(\theta^\prime, \lambda^\prime) \in \mc S_\theta} \norm{(\theta^\prime -\theta,\lambda^\prime-\lambda)} \label{eq:ps:theta}
\end{align}
and let $(\theta^{\prime\star},\lambda^{\prime\star}) \in \mc S_\theta$ denote the (unique) optimizer of \eqref{eq:ps:theta}. Note that $(VB^\mathsf{T}\theta^{\prime\star},\lambda^{\prime\star}) \in \mc S_\eta$, by Cauchy–Schwarz inequality, it follows that
\begin{align*}
\norm{T_\eta(\theta^{\prime\star} - \theta,\lambda^{\prime\star}-\lambda)}\leq \norm{T_\eta} \norm{(\theta^{\prime\star} - \theta,\lambda^{\prime\star}-\lambda)},
\end{align*}
i.e., any optimizer of \eqref{eq:ps:theta} can be used to upper bound \eqref{eq:ps:eta} in terms of \eqref{eq:ps:theta}. It immediately follows that the first inequality in \eqref{eq:normbound} holds with $\underline{\kappa}=\norm{T_\eta}^{-1} \in \mathbb{R}_{>0}$. Next, let 
\begin{align*}
    \sigma_{\mathbbl{1}}\coloneqq \min\nolimits_{\theta \perp \mathbbl{1}_n, \norm{\theta}=1} \norm{VB^\mathsf{T}\theta} \in \mathbb{R}_{>0},
\end{align*}
and decompose $\eta^\prime= VB^\mathsf{T} \beta^\prime + z^\prime \in \mathbb{R}^e$ into $\beta^\prime \in \mathbb{R}^n$ and $z^\prime \perp \Ima(VB^\mathsf{T})$, i.e., $z^\prime \in \ker(BV)$. Then, $\norm{T_\eta ({\theta},\lambda)}_{\mc S_\eta}$ can be written as %
\begin{align}
    &\min_{(\beta^\prime,z^\prime,\lambda^\prime): (VB^\mathsf{T}\beta^\prime+z^\prime,\lambda^\prime)\in\mc S_\eta} \norm{(VB^\mathsf{T}(\beta^\prime-\theta)+z^\prime,\lambda^\prime-\lambda)}, \nonumber\\
    &= \min_{(\beta^\prime,\lambda^\prime): (VB^\mathsf{T}\beta^\prime,\lambda^\prime)\in\mc S_\eta} \norm{(VB^\mathsf{T}(\beta^\prime-\theta),\lambda^\prime-\lambda)}, \label{eq:ps:betaz}
\end{align}
where we used that
\begin{align*}
    \norm{VB^\mathsf{T} y+z^\prime} = \sqrt{y^\mathsf{T} L y+2 y^\mathsf{T} BV z^\prime + (z^\prime)^\mathsf{T} z^\prime},
\end{align*}
with $y=\beta^\prime-\theta$ and $BV z^\prime = \mathbbl{0}_n$ to conclude that the minimum of \eqref{eq:ps:betaz} is attained at $z^\prime = \mathbbl{0}_e$. In other words, for any $(\beta^\prime, \lambda^\prime)$ such that $(VB^\mathsf{T}\beta^\prime,\lambda^\prime)\in\mc S_\eta$ we obtain
\begin{align*}
    \min_{\beta^\prime,\lambda^\prime} \norm{\!\begin{bmatrix}VB^\mathsf{T}(\beta^\prime-\theta) \\ \lambda^\prime-\lambda\end{bmatrix}\!} &\geq \min_{\beta^\prime,\lambda^\prime} \norm{\!\begin{bmatrix} \sigma_\mathbbl{1} I_n & \mathbbl{0}_{n \times 2n} \\ \mathbbl{0}_{2n \times n} & I_{2n} \end{bmatrix} \begin{bmatrix} \beta^\prime-\theta \\ \lambda^\prime-\lambda \end{bmatrix}\!} \\
     &\geq \min\{\sigma_{\mathbbl{1}}, 1\} \min_{\beta^\prime,\lambda^\prime} \norm{\! \begin{bmatrix} \beta^\prime-\theta \\ \lambda^\prime-\lambda \end{bmatrix}\!} 
\end{align*}
and \eqref{eq:normbound} holds with $\overline{\kappa}= \frac{1}{\min\{\sigma_\mathbbl{1}, 1\}} \in \mathbb{R}_{>0}$.
By Theorem~\ref{thm:edgeconvergence}, \eqref{eq:primaldualedge} is GAS on $\mathbb{R}^{e} \times \mathbb{R}^{2n}_{\geq0}$ with respect to $\mc S_\eta$. In other words, there exists $\chi_\eta \in \mathscr{KL}$ such that
\begin{align*}
\norm{\varphi_\eta(t,T_\eta \xi_0)}_{\mc S_\eta} \leq \chi_\eta({\norm{T_\eta \xi_0}_{\mc S_\eta}, t})
\end{align*}
for all $\xi_0 \in \mathbb{R}^n \times \mathbb{R}^{2n}_{>0}$ and all $t \in \mathbb{R}_{\geq 0}$. Using Lemma~\ref{lem:coincide}, \eqref{eq:normbound}, and $\chi_\eta({\norm{T_\eta \xi_0}_{\mc S_\eta}, t}) \leq \chi_\eta({\underline{\kappa}^{-1}\norm{T_\eta \xi_0}_{\mc S_\theta}, t})$, we obtain
\begin{align*}
    \norm{\varphi_\theta(t,\xi_0)}_{\mc S_\theta} \leq \overline{\kappa}  \norm{\varphi_\eta(t,T_\eta \xi_0)}_{\mc S_\eta} \leq \overline{\kappa} \chi_\eta({\underline{\kappa}^{-1} \norm{T_\eta \xi_0}_{\mc S_\theta}, t}     )
\end{align*}
for all $\xi_0 \in \mathbb{R}^n \times \mathbb{R}^{2n}_{>0}$ and all $t \in \mathbb{R}_{\geq 0}$. In other words, \eqref{eq:plimdroop} is globally asymptotically stable on $\mathbb{R}^n \times \mathbb{R}^{2n}_{\geq 0}$ with respect to the set  $\mathcal{S}_{\theta}$. Finally, we show that $\lim_{t \to \infty} \omega(t) = \mathbbl{1}_n \omega_{s}$. According to Theorem~\ref{thm:edgeconvergence}, any pair $\left(\eta^\star, \lambda^\star\right)$ converges to a KKT point $\left(\eta^\star, \lambda^\star \right) \in \mc S_\eta$, i.e., $\lim_{t \to \infty} \eta(t) = \eta^\star$ and $\lim_{t \to \infty} \ddt \eta=\mathbbl{0}_e$. Using $\eta = V B^\mathsf{T} \theta$, we obtain
\begin{align*}
   \lim_{t \to \infty} \ddt V B^\mathsf{T} \theta(t) = \lim_{t \to \infty} V B^\mathsf{T} \ddt \theta(t) = \lim_{t \to \infty} VB^\mathsf{T} \omega(t) = 0
\end{align*}
and $\lim_{t \to \infty} \omega(t) = \mathbbl{1}_n \omega_{s}$ follows $\ker(B^\mathsf{T})=\vspan(\mathbbl{1}_n)$.
\hfill $\square$


\section{Application to Frequency Dynamics of Multi-Converter Power Systems}\label{sec:application}
\subsection{Converter and network model}
In the context of the multi-converter power systems, the nodal variable $\theta \in \mathbb{R}^n$ models voltage phase angles of each power converter and $P_L$ models local load of each converter. Throughout this work, we assume that the power network is lossless and modeled through a Kron-reduced graph~\cite{DB2013} and its so-called dc power flow \eqref{eq:dcpfeq}. We note that the CFP \eqref{eq:deCFP}  corresponds to an optimal dc power flow problem that seeks voltage phase angles $\theta^\star \in \mathbb{R}^n$ that minimize deviations from the power setpoint $P^\star \in \mathbb{R}^n$ subject to converter power limits $P_\ell \in \mathbb{R}^n$ and $P_u\in \mathbb{R}^n$ and supplying the load $P_L$.

\subsection{Review of power-limiting droop control}
Grid-forming power-limiting droop control~\cite{PL2006,DLK2019} as shown in Fig.~\ref{fig:plim} uses a measurement of the converter power injection $P_i \in \mathbb{R}$ to determine the frequency $\omega_i = \ddt \theta_i \in \mathbb{R}$ (relative to the nominal frequency $\omega_0$) of the ac voltage imposed by the converter at its bus $i \in \mc N$. Notably, power-limiting droop control combines the widely studied (proportional) $P-f$ droop control~\cite{CDA1993,RLB2012,SDB2013} with (nonlinear) PI controls that aim to maintain converter power injection $P_i \in \mathbb{R}$ within lower and upper limits $P_{\ell, i} \in \mathbb{R}$ and $P_{u, i} \in \mathbb{R}$. 
\begin{figure}[htbp]
    \begin{center}
        \includegraphics[width=1\columnwidth]{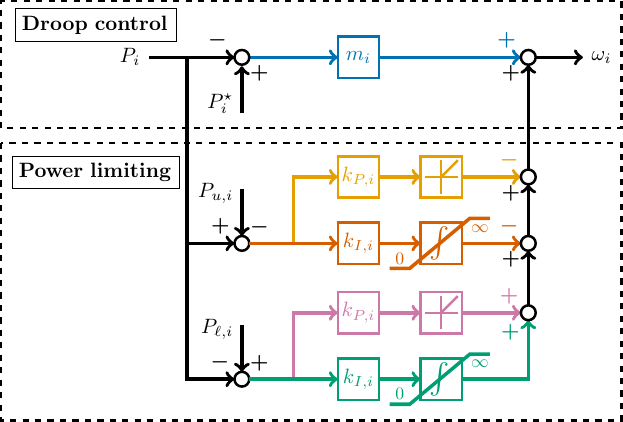}
        \caption{Grid-forming power-limiting droop control combines droop control with nonlinear proportional-integral controls for power-limiting that activate for power-limiting when a power limit is exceeded.
        \label{fig:plim}}
    \end{center}
\end{figure}
When no power limit is active (i.e., $P_{\ell,i} < P_i < P_{u,i}$), power-limiting droop control reduces to well-known (proportional) $P-f$ droop control~\cite{CDA1993} that controls the converter power injection~\cite[Sec.~IV-C]{RLB2012} by adjusting the frequency $\omega_i$ in proportion to the deviation of $P_i\in\mathbb{R}$ from the converter power setpoint $P^\star_i\in\mathbb{R}$ to enable parallel operation~\cite{CDA1993} of grid-forming converters.

We emphasize that, in general, the load $P_L \in \mathbb{R}$ is not known, and the sum of the power setpoints does not match the load (i.e., $\sum_{i\in\mc N} P^\star_i \neq \sum_{i\in\mc N} P_{L,i}$). Moreover, the exact topology and edge weights of the graph $\mc G$ are generally not known. In this setting, the control objective is to render a synchronous solution (i.e., $\omega_i=\omega_j$ for all $(i,j) \in \mc N \times \mc N$) stable for any connected graph $\mc G$ while sharing any additional load between the converters according to the ratio of droop coefficients $m_i \in \mathbb{R}_{>0}$ (see, e.g., \cite{SDB2013}). While $P-f$ droop control achieves these objectives under mild assumptions~\cite{SDB2013, SOA2014}, it does not account for the converter power limits $P_{\ell,i} < P_i < P_{u,i}$. A common heuristic used to include power limits employs proportional-integral (PI) $P-f$ droop with proportional and integral gains $k_{P,i} \in \mathbb{R}_{>0}$ and $k_{I,i} \in \mathbb{R}_{>0}$ that activate when a power limit is reached~\cite[Fig.~4]{DLK2019}. 

For example, if a converter reaches or exceeds its upper power limit (i.e., $P_i \geq P_{u,i}$), then power-limiting droop control depicted in Fig.~\ref{fig:plim} will reduce the frequency in proportion to the constraint violation and its integral. Due to the nature of integral control, this control should intuitively control the converter to an operating point within its power limits asymptotically. In the application context, the main contribution of this work is to use Theorem~\ref{thm:GASpowerlimit} to show that power-limiting droop control renders the overall converter-based system stable with respect to a synchronous solution within the converter power limits. To this end, we note that Assumption~\ref{assum:feas} and Assumption~\ref{assum:setpoint} formalize assumptions that are implicitly made in the literature on grid-forming control.

\subsection{Multi-converter power system frequency dynamics as projected dynamical system}
We begin by formulating the frequency dynamics of a power system comprised of converters using power-limiting droop control as a projected dynamical system. Using Definition~\ref{def:tangentcone} and Definition~\ref{def:projection} the frequency dynamics of a power system consisting of converters using power-limiting droop control is equivalent to the projected dynamical system \eqref{eq:plimdroopprime}. 
\begin{subequations}\label{eq:plimdroopprime}
    \begin{align}
       \ddt \theta_i =& m_i (P^\star_i-P_i)\! -\! k_{P,i} \Pi_{{\mathbb{R}}_{\geq 0}}(P_i-P_{u,i}) - \mu_{u,i}\\ &+ k_{P,i} \Pi_{{\mathbb{R}}_{\geq 0} }(P_{\ell,i} - P_i) + \mu_{\ell,i}, \nonumber\\
    \ddt\mu_{\ell,i} =&  \Pi_{\mathcal{T}_{\mu_{\ell,i}} \mathbb{R}_{\geq 0}} \left(k_{I,i} (P_{\ell,i} - P_i)\right),\\
    \ddt\mu_{u,i} =&  \Pi_{\mathcal{T}_{\mu_{u,i}} \mathbb{R}_{\geq 0}} \left(k_{I,i} (P_i - P_{u,i})\right),
    \end{align}
\end{subequations}

In this case, the controller states $\theta_i$,  $\mu_{u,i} \in \mathbb{R}_{\geq 0}$, and $\mu_{\ell,i} \in \mathbb{R}_{\geq 0}$ that correspond to the ac voltage phase angles and integral of the upper and lower power limit violations. Moreover, we define the ac voltage frequency deviation $\omega_i = \ddt \theta_i \in \mathbb{R}$ from the nominal frequency $\omega_0$. To simplify our analysis and notation, we introduce the following preliminary result.
\begin{lemma}[\textbf{Scaled scalar projection}]\label{lem:proj}
        Given a constant $a \in \mathbb{R}_{>0}$ and scalars $v \in \mathbb{R}$ and $x \in \mathbb{R}$, it holds that $\Pi_{\mathcal{T}_{x} \mathbb{R}_{\geq 0}}(a v) = a \Pi_{\mathcal{T}_{x} \mathbb{R}_{\geq 0}}(v)$.
    \end{lemma}
    \begin{proof}
        Using~\cite[Prop.~5.3.5]{jean}, $\Pi_{\mathcal{T}_{x} \mathbb{R}_{\geq 0}}(v)$ can be expressed as $\Pi_{\mathcal{T}_{x}\mathcal{C}}(v) = \lim_{\delta \to 0} \frac{1}{\delta} (\Pi_{\mathcal{C}}(x+\delta v) - x)$. Then, for $x \in \partial(\mathcal{C})$ it holds that $\Pi_{\mathcal{T}_{x} \mathbb{R}_{\geq 0}}(a v) = \lim_{\delta \to 0} \frac{1}{\delta}(\Pi_{\mathbb{R}_{\geq 0}}(x+\delta a v) - x)$. Letting $\delta^\prime=a \delta$ results in
        $\Pi_{\mathcal{T}_{x} \mathbb{R}_{\geq 0}}(a v) =a \lim_{\delta^\prime \to 0} \frac{1}{\delta^\prime} (\Pi_{\mathbb{R}_{\geq 0}}(x+\delta^\prime v) - x)=a \Pi_{\mathcal{T}_{x} \mathbb{R}_{\geq 0}}(v)$.
    \end{proof}

Using Lemma~\ref{lem:proj}, we can rewrite \eqref{eq:plimdroopprime} using the change of variables $\sqrt{k_{I,i}} \lambda_{\ell,i} = \mu_{\ell,i}$, and $\sqrt{k_{I,i}} \lambda_{u,i} = \mu_{u,i}$ as \eqref{eq:plimdroopsc}. We emphasize that this model assumes that the load $P_L$, power setpoints $P^\star$, and power limits $P_\ell$ and $P_u$ are constant. This assumption is satisfied on the time-scales of interest for studying frequency stability of converter-dominated power systems. Extensions to time-varying loads, setpoints, and power limits are seen as an interesting area for future work. 

\subsection{Frequency synchronization and active constraints}\label{subsec:dynconstr}
The synchronous frequency $\omega_s \in \mathbb{R}$ is widely used in control and operation of power systems. To analyze the steady-state frequency, we use properties of the set of KKT points to establish that \eqref{eq:plimdroop} achieves frequency synchronization under constraints and establish that the synchronous frequency deviation is a function of the active sets $\mc I_u$ and $\mc I_\ell$, total load $\sum_{i \in \mc N}  P_{L,i}$, total power dispatch $\sum_{i \in \mc N} P^\star_i$, and droop coefficients, but does not depend on the graph $\mc G$ or the control gains of the power-limiting PI controls.

\begin{theorem}[\textbf{Frequency synchronization}]\label{th:syncfreq}
    Consider $P_\ell$, $P_u$, $P_L$, and $P^\star$ such that Assumption~\ref{assum:feas} and Assumption~\ref{assum:setpoint} hold. One of the following holds
    \begin{enumerate}[label=(\roman*)]
        \item $\sum_{i \in \mc N}  P_{L,i} < \sum_{i \in \mc N} P^\star_i$ and \[\omega_s\!=\!\frac{\sum_{i \notin \mathcal{I}_\ell} P_i^\star+\sum_{i \in \mathcal{I}_\ell} P_{\ell, i}-\sum_{i \in \mathcal{N}} P_{L, i}}{\sum_{i \notin \mathcal{I}_\ell} m_i^{-1}} > 0,\] \label{th:syncfreq:lower}
        \item   $\sum_{i \in \mc N}  P_{L,i} = \sum_{i \in \mc N} P^\star_i$ and $\omega_s=0$, \label{th:syncfreq:equal}
        \item $\sum_{i \in \mc N}  P_{L,i} > \sum_{i \in \mc N} P^\star_i$ and 
        \[\omega_s\!=\!\frac{\sum_{i \notin \mathcal{I}_u} P_i^\star+\sum_{i \in \mathcal{I}_u} P_{u, i}-\sum_{i \in \mathcal{N}} P_{L, i}}{\sum_{i \notin \mathcal{I}_u} m_i^{-1}} < 0.\] \label{th:syncfreq:higher}
    \end{enumerate}
\end{theorem}
\begin{proof} By Theorem~\ref{thm:GASpowerlimit}, there exists $(\theta^\star,\lambda^\star) \in \mc S_\theta$ such that
\begin{align*}
    \lim_{t \to \infty} \omega(t) =  M \left(P^\star - P_{L} - L \theta^\star \right) - (\alpha \otimes K_I) \lambda^\star = \mathbbl{1}_n \omega_{s}.
\end{align*}
Where we have used \eqref{eq:plimdroop:theta} and the fact that $\Pi_{\mathbb{R}^{2n}_{\geq 0}}(g(L\theta^\star))=\mathbbl{0}_{2n}$ for all $(\theta^\star, \lambda^\star) \in \mc S_\theta$. Using $P=L\theta + P_L$ it follows that
\begin{align}\label{eq:omegas}
   m^{-1}_i \omega_{s} = P^\star_i - P_i + m_i^{-1} \sqrt{k_i}(\lambda^\star_{\ell,i} - \lambda^\star_{u,i}), \; \forall i \in \mc N.
\end{align}
Since $\lambda^\star_{u,i}=\lambda^\star_{\ell,i}=0$ for all $i\notin \mc I_u \cup \mc I_\ell$, we obtain
\begin{align}\label{eq:freqbalance}
     \sum\nolimits_{i\notin \mc I_u \cup \mc I_\ell}  m^{-1}_i \omega_{s} = \sum\nolimits_{i\notin \mc I_u \cup \mc I_\ell} P^\star_i - \sum\nolimits_{i\notin \mc I_u \cup \mc I_\ell}  P_i.
\end{align}
Moreover, considering $\mathbbl{1}^\mathsf{T}_n (L \theta + P_L) = \mathbbl{1}^\mathsf{T}_n P_L$ it follows that $\sum_{i \in \mc N} P_i = \sum_{i \in \mc N} P_{L,i}$ and
\begin{align*}
    \sum\nolimits_{i \notin \mc I_u \cup \mc I_\ell} P_i + \sum\nolimits_{i \in \mc I_\ell} P_{\ell,i} + \sum\nolimits_{i \in \mc I_u} P_{u,i}   = \sum\nolimits_{i \in \mc N} P_{L,i}.
\end{align*}
Solving for $\sum_{i \notin \mc I_u \cup \mc I_\ell} P_i$ and substituting into \eqref{eq:freqbalance}, results in
\begin{align*}
    \sum_{i \notin \mathcal{I}_{u} \cup \mathcal{I}_{\ell}}\! m_{i}^{-1} \omega_{s} \!= \!\sum_{i \notin \mathcal{I}_{u} \cup \mathcal{I}_{\ell}}\! P^\star_i + \sum_{i \in \mathcal{I}_{u}}\! P_{u, i} + \sum_{i \in \mathcal{I}_{\ell}}\! P_{\ell, i} - \sum_{i \in \mc N}\! P_{L,i}
\end{align*}
and
\begin{align}
    &\omega_s =\frac{\sum_{i \notin \mathcal{I}_{u} \cup \mathcal{I}_{\ell}} \! P_i^\star+ \sum_{i \in \mathcal{I}_u}\! P_{u, i}+\sum_{i \in \mathcal{I}_\ell} \! P_{\ell, i}-\sum_{i \in \mathcal{N}}\! P_{L, i}}{\sum_{i \notin \mathcal{I}_{u} \cup \mathcal{I}_{\ell}}\! m_i^{-1}}. \label{eq:omega_iliu}
\end{align}

To show \ref{th:syncfreq:lower}, consider $\mathcal{I}_{\ell} \neq \emptyset$. Then, by Proposition~\ref{prop:mut:activesets}, $\mathcal{I}_{u} = \emptyset$. Moreover, (i) $\lambda_{\ell, i} \geq 0$ by dual feasibility, (ii) $\lambda_{u, i} = 0$ by complementary slackness, and (iii) $P^\star_i - P_{\ell, i} >0$ for all $i \in \mc N$ by Assumption~\ref{assum:setpoint}. Then \eqref{eq:omegas} results in
\begin{align*}
    \omega_s = \omega_i = m_i (P^\star_i - P_{\ell, i}) + \sqrt{k_i} (\lambda_{\ell, i} - \lambda_{u,i}) >0, \; \forall i \in \mathcal{I}_{\ell},
\end{align*}
i.e., $\omega_s>0$ holds. Moreover, \eqref{eq:omega_iliu} and $\mc I_u=\emptyset$ imply that
\begin{align*}
    \sum\nolimits_{i \notin \mathcal{I}_{\ell}}  P_i^\star+\sum\nolimits_{i \in \mathcal{I}_\ell}  P_{\ell, i}-\sum\nolimits_{i \in \mathcal{N}} P_{L, i} >0.
\end{align*}
Additionally, Assumption~\ref{assum:feas} implies $\sum_{i \in \mathcal{I}_\ell} \! P_{\ell, i} < \sum_{i \in \mathcal{I}_\ell} \! P^\star_i$, i.e.,
\begin{align*}
    &\sum_{i \in \mc N} \! P_i^\star-\sum_{i \in \mathcal{N}}\! P_{L, i} > \sum_{i \notin \mathcal{I}_{\ell}} \! P_i^\star+\sum_{i \in \mathcal{I}_\ell} \! P_{\ell, i}-\sum_{i \in \mathcal{N}}\! P_{L, i} >0,
\end{align*}
i.e., $\sum_{i \in \mathcal{N}}\! P_{L, i}<\sum_{i \in \mc N} \! P_i^\star$. Assuming from $\mathcal{I}_{u} \neq \emptyset$ and applying the same steps used for showing \ref{th:syncfreq:lower} establishes \ref{th:syncfreq:higher}.
%

It remains to show \ref{th:syncfreq:equal}. To this end, assume that $\mathcal{I}_{\ell} = \emptyset$. Using $\sum_{i \in \mc N} P_{L,i} = \sum_{i \in \mc N} P^\star_i$, the synchronous frequency reduces to 
\begin{align*}
    \omega_{s} = \frac{\Sigma_{i \in \mathcal{I}_{u}}P_{u, i} - P^\star_{i}}{\Sigma_{i \notin \mathcal{I}_{u}} m_{i}^{-1}} > 0.
\end{align*}
However, (i) $\lambda^\star_{u,i}\geq 0$ by primal feasibility, (ii) $\lambda^\star_{\ell,i}=0$, by complementary slackness, and (iii) $P^\star_i - P_{u, i}<0$ for all $i \in \mc N$ by Assumption~\ref{assum:setpoint}. This results in 
\begin{align*}
\omega_s = \omega_i = m_i\left(P^\star_i - P_{u, i}\right) - \sqrt{k_i}\lambda_{u, i} < 0, \quad \forall i \in \mathcal{I}_{u},
\end{align*}
and it follows that $\mathcal{I}_{u} = \mathcal{I}_{\ell} = \emptyset$. Following the same steps for $\mathcal{I}_u = \emptyset$ and $\sum_{i \in \mc N} P_{L,i} = \sum_{i \in \mc N} P^\star_i$ shows that $\mathcal{I}_\ell = \emptyset$ has to hold. The Theorem follows by noting that either $\mc I_u=\emptyset$ or $\mc I_\ell=\emptyset$ by Proposition~\ref{prop:mut:activesets}.
\end{proof} 

Finally, we show that the mismatch between the total load $\sum_{i \in \mc N}  P_{L,i}$ and total power dispatch $\sum_{i \in \mc N} P^\star_i$ determines which converter limits are active upon convergence.
\begin{proposition}[\textbf{Active constraint set}]\label{prop:activesets}
    Consider $P_\ell$, $P_u$, $P_L$, and $P^\star$ such that Assumption~\ref{assum:feas} and Assumption~\ref{assum:setpoint} hold.  Then, for all $(\theta,\lambda) \in \mathcal{S}_{\theta}$ it holds that 
    \begin{enumerate}[label=(\roman*)]
    \item $\mathcal{I}_{\ell} \neq \emptyset$ implies $\mathcal{I}_{u} = \emptyset$ and $\sum_{i \in \mc N} P_{L,i} < \sum_{i \in \mc N} P^\star_i$, \label{prop:activesets:ellu}
    \item $\mathcal{I}_{u} \neq \emptyset$ implies $\mathcal{I}_{\ell} = \emptyset$ and $\sum_{i \in \mc N} P_{L,i} > \sum_{i \in \mc N} P^\star_i$, \label{prop:activesets:uell}
    \item $\sum_{i \in \mc N}  P_{L,i} < \sum_{i \in \mc N} P^\star_i$ implies $\mathcal{I}_{u} = \emptyset$, \label{prop:activesets:lowload}
    \item $\sum_{i \in \mc N} P_{L,i} > \sum_{i \in \mc N} P^\star_i$ implies $\mathcal{I}_{\ell} = \emptyset$. \label{prop:activesets:highload}
    \end{enumerate}
\end{proposition}
\begin{proof} Statement~\ref{prop:activesets:ellu} and statement~\ref{prop:activesets:uell} are direct consequences of Proposition~\ref{prop:mut:activesets} and the proof of Theorem~\ref{th:syncfreq}. Statement~\ref{prop:activesets:lowload} can be shown by contradiction. In particular, $\sum_{i \in \mc N} P_{L,i} < \sum_{i \in \mc N} P^\star_i$ implies $\omega_s>0$. However, per the proof of Theorem~\ref{th:syncfreq}, if there exists $i \in \mc I_u \neq \emptyset$, then $\omega_s<0$. Thus, $\mc I_u = \emptyset$ has to hold. Statement~\ref{prop:activesets:highload} follows by noting that $\sum_{i \in \mc N} P_{L,i} > \sum_{i \in \mc N} P^\star_i$ implies $\omega_s<0$. However, per the proof of Theorem~\ref{th:syncfreq}, if there exists $i \in \mc I_\ell \neq \emptyset$, then $\omega_s>0$. Thus, $\mc I_\ell = \emptyset$ has to hold. \end{proof}

\subsection{Discussion}
Theorem~\ref{th:syncfreq} recovers and extends the well known results for (proportional) $P-f$ droop control, i.e., if no converter is operating at a power limit (i.e., $\mc I_u = \mc I_\ell = \emptyset$), then the steady-state frequency deviation is determined by the droop coefficients $m_i \in \mathbb{R}_{>0}$ and the mismatch $\sum_{i \in \mc N} P^\star_i - P_{L,i}$ between the total power dispatch and load.

Moreover, if the total load $\sum_{i \in \mc N}  P_{L,i}$ is smaller than the total power dispatch $\sum_{i \in \mc N} P^\star_i$, then converters can only be at their lower power limit (i.e., $\mc I_u=\emptyset$) and the synchronous frequency is determined by the sum of the droop coefficients and sum of the power setpoints of converters not at the lower limit (i.e., $i \notin \mc I_\ell$), the total load, and the sum of the lower power limits of the converters at the lower limit (i.e., $i \in \mc I_\ell$). In contrast, if the total load $\sum_{i \in \mc N}  P_{L,i}$ is larger than the total power dispatch  $\sum_{i \in \mc N}  P^\star_i$, then converters can only be at their upper power limit (i.e., $\mc I_\ell=\emptyset$) and the synchronous frequency is determined by the sum of the droop coefficients and sum of the power setpoints of converters not at the upper limit (i.e., $i \notin \mc I_u$), the total load, and the sum of the upper power limits of the converters at the upper limit (i.e., $i \in \mc I_u$). Moreover, we note that if $\sum_{i \in \mc N}  P_{L,i}$ is smaller (larger) than the total power dispatch  $\sum_{i \in \mc N}  P^\star_i$, then the synchronous frequency is larger (smaller) than nominal frequency $\omega_0$. The remainder of the manuscript will focus on proving the aforementioned results.

Finally, we note that by Corollary~\ref{corr:radial}, the multi-converter power system frequency dynamics \eqref{eq:plimdroopprime} are guaranteed to converge to a uniquely determined voltage phase angle difference $\eta$ across each transmission line if the power network has a radial topology. In contrast, for a meshed network, the optimal voltage phase angle difference $\eta$ across each transmission line may not be unique. In other words, different solutions of the optimal dc power flow problem \eqref{eq:deCFP} may correspond to different line loading.

\begin{figure*}[htbp]
    \centering
    \vspace{-4ex}
        \includegraphics[width=0.95\linewidth]{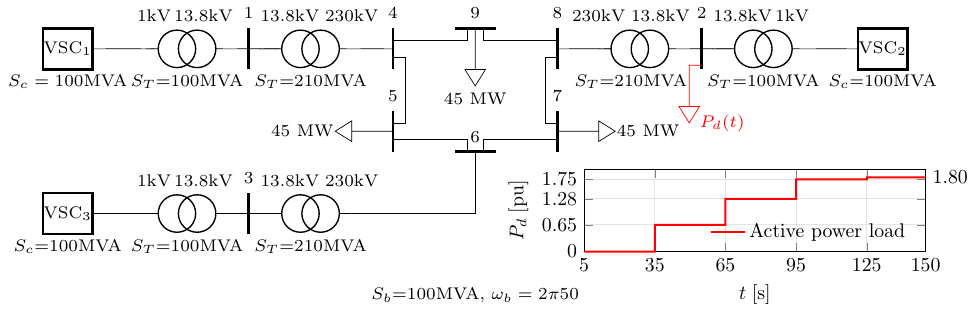}
        \caption{IEEE 9-bus test case system with three two-level voltage source converters and constant (black) and dynamic (red) loads.        \label{fig:9bus_system}}
    \begin{center}
        \includegraphics[width=1\linewidth]{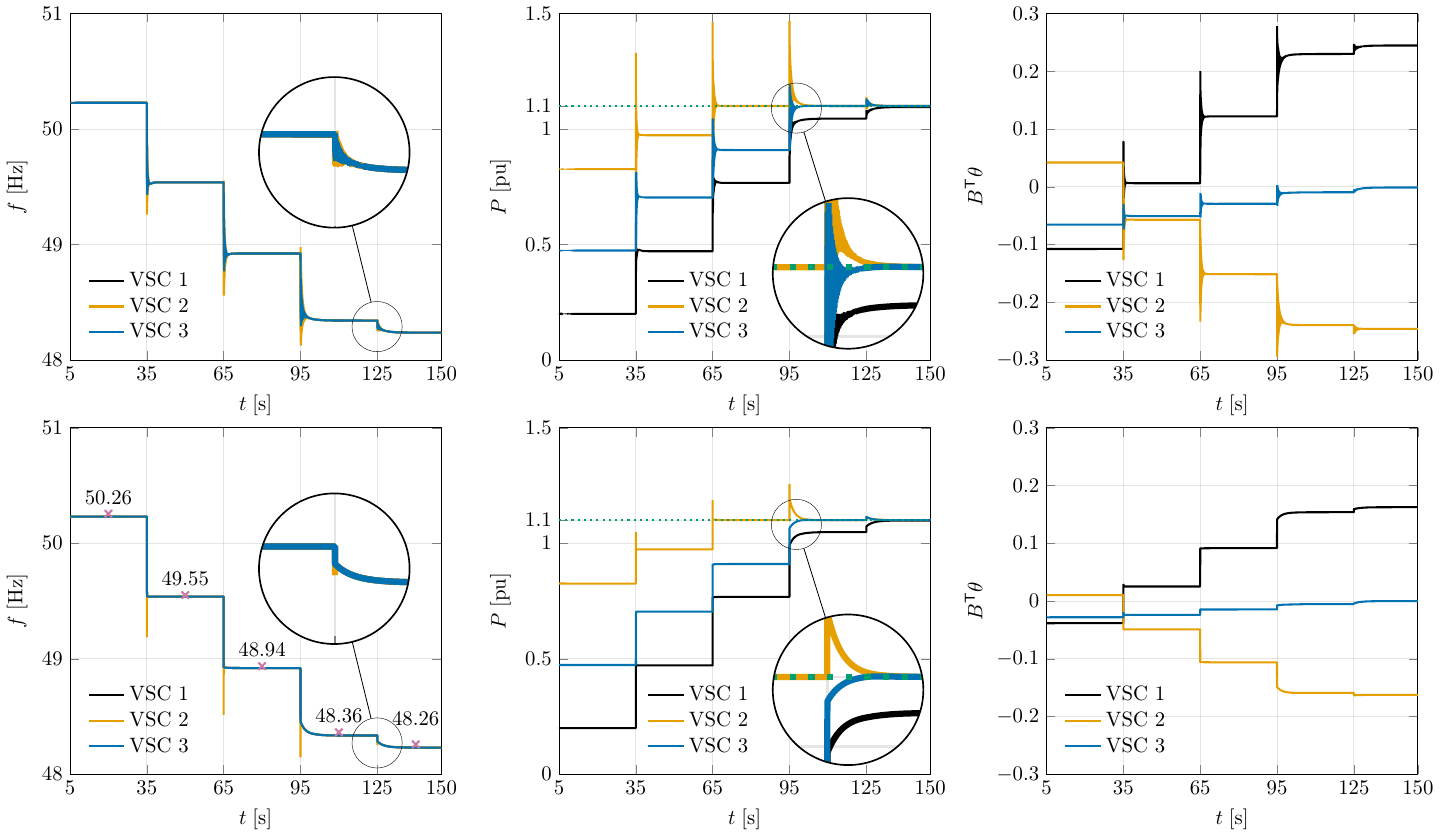}
        \caption{Results of an EMT simulation (top row) and the reduced-order model (bottom row) for the IEEE 9-bus system depicted in Figure~\ref{fig:9bus_system}. The green line indicates the upper power limit of each VSC and the pink markers indiciate the frequency deviation predicted by Theorem~\ref{th:syncfreq} using the parameters in Table~\ref{table:parameters}.\label{fig:freqresponse}}
    \end{center}
\end{figure*}

\section{Numerical Case Study}\label{sec:numerical}
To validate the reduced-order model \eqref{eq:plimdroopprime} of the multi-converter power system and illustrate the analytical results, we use an Electromagnetic transient (EMT) simulation of the IEEE 9-bus system (see Figure~\ref{fig:9bus_system}) in which synchronous generators have been replaced with two-level voltage source converters controlled by power-limiting droop control.
\subsection{Power System Model}
The IEEE 9-bus system with three power converters is shown in Figure~\ref{fig:9bus_system}. Specifically, we use an average model of two-level voltage source converters (VSCs) with $LC$ output filter and standard cascaded inner voltage and current loops. The reader is referred to \cite{TGA+2020} for details on the transformer parameters, converter parameters, and control gains of the inner control loops. The converter rating, power setpoints, power limits, and control gains for power-limiting droop control used in this work are given in Table~\ref{table:parameters}.  
\begin{table}[b!]
    \centering
    \caption{Converter and control parameters for power-limiting droop control.}
    \label{table:parameters}
    \scriptsize
    \resizebox{\columnwidth}{!}{
    \begin{tabular}{|c|c|c|c|c|c|c|}
    \hline
    \multirow{2}{*}{VSC} & \multicolumn{3}{c|}{Power [MW]} & \multicolumn{3}{c|}{Control gains [pu]} \\
    \cline{2-7} 
    \rule{0pt}{2.1ex}
    & $P^\star$ & $P_\ell$ & $P_u$ & $m_p$ & $k_P$ & $k_I$ \\
    \hline
    1 & 25 MW & 20 MW & 110 MW & 4.17\% & 0.0048 & 0.0637 \\
    \hline
    2 & 87.5 MW & 20 MW & 110 MW & 9.38\% & 0.0048 & 0.0637 \\
    \hline
    3 & 55 MW & 20 MW & 110 MW & 6\% & 0.0048 & 0.0637 \\
    \hline
    \end{tabular}
    }
\end{table}

In addition to the base load, the dynamic load shown in Figure~\ref{fig:9bus_system} is introduced to create overload conditions for VSCs. The results of the EMT simulation are compared to the reduced-order model \eqref{eq:plimdroopprime}. 
We note that the reduced-order model \eqref{eq:plimdroopprime} uses a Kron-reduced~\cite{DB2013} network model with three buses and assumes decoupling of active power and frequency from reactive power and voltage. Moreover, the reduced-order model neglects fast dynamics (i.e., inner control loops, circuit dynamics, transmission line dynamics).

\subsection{Simulation Results and Discussion}
Simulation results are shown in Figure~\ref{fig:freqresponse}. EMT simulation results are shown in the top row and simulation results obtained using the reduced-order model \eqref{eq:plimdroopprime} are shown in the bottom row. Finally, the pink markers in the bottom row indicate the frequency deviation predicted by Theorem~\ref{th:syncfreq}.

From $t=5~\mathrm{s}$ to $t=35~\mathrm{s}$, the system is in steady-state with a frequency deviation because the total load is below the total dispatch, i.e., $\sum  P_{L,i} < \sum P^\star_i$. After $t=35~\mathrm{s}$, the dynamic load increases stepwise such that $\sum P_{L,i} > \sum P^\star_i$. Therefore, from $t=35~\mathrm{s}$ to $t=65~\mathrm{s}$ load increments result in a decrease in frequency. Notably, for $t < 65~\mathrm{s}$, the total load is below the sum of the upper power limits and the power-limiting does not activate. Thus, conventional power-sharing occurs as prescribed by the droop coefficients. After applying the load increment at $t=65~\mathrm{s}$, VSC 2 becomes overloaded and the additional load is shared between the two VSCs operating below their upper power limit indicated by the green line in Figure~\ref{fig:freqresponse}. It is worthwhile to note that, in this setting, VSC 2 is more prone to overload because its power setpoint $P^\star_2$ is close to its upper limit $P_{u,2}$. Similarly, the consecutive subsequent load increases overload VSC 3 and VSC 1, respectively. We note that the reduced-order model captures the average transient responses of the EMT model as illustrated by the zoomed in frequency and power at $t = 125~\mathrm{s}$ and power at $t=95~\mathrm{s}$. 

Moreover, to validate the analytical results for the synchronous frequency $\omega_s$, the frequency deviation predicted by Theorem~\ref{th:syncfreq} for the aforementioned load scenarios have been computed and are indicated using pink markers in Figure~\ref{fig:freqresponse}. We observe that the  analytical results closely match the EMT simulation and results, and obtained using the reduced-order model \eqref{eq:plimdroopprime}.

Finally, the last column in Figure~\ref{fig:freqresponse} shows the primal variables in edge coordinates for both simulations. The results show that the EMT simulation and reduced-order model \eqref{eq:plimdroopprime} exhibit the same qualitative response upon frequency synchronization and convergence. At the same time, we observe that the angle differences are smaller in the reduced-order model as compared to the EMT simulation. These differences arise from line neglecting losses and the linearization of ac power flow equation at nominal voltage magnitude and zero angle differences. In contrast, the frequency and the power injection of the VSCs closely match between the EMT simulation, reduced-order model, and analytical results.

\section{Conclusions and Outlook}\label{sec:conclusion}
In this paper, we studied a constrained network flow problem that aims to minimize the deviation of nodal network injections from given references under injection limits. Applying standard primal-dual dynamics to the constrained flow problem results in dynamics that cannot be implemented using only local measurements. Instead, we investigated networked dynamics that leverage measurements of the network injections to solve the constrained flow problem. First, we showed that the networked dynamics are equivalent to primal-dual dynamics of the constrained flow problem in edge coordinates. Then, we established that the networked dynamics are asymptotically stable with respect to the set of the KKT points of the constrained flow problem in the original nodal coordinates. Leveraging our theoretical results, we showed that, under mild assumptions, the frequency dynamics of power network comprised of power converters using power-limiting droop control are globally asymptotically stable with respect to set of KKT points of a constrained dc power flow problem. Furthermore, we established that (i) the converters synchronize to a common synchronous frequency, and (ii) exhibit power-sharing properties similar to conventional unconstrained droop control. Specifically, we analyzed the impact of power limits on the synchronous steady-state frequency of the power system and characterized the synchronous frequency in terms of total load, sum of converter power setpoints, and converter droop coefficients. Particularly, converters share additional load according to their droop coefficients up to their power limit. While these results are encouraging, future work should consider a wider range of constraints. In the power system context, constraints beyond active power limits (e.g., current limits, dc voltage limits) are seen as interesting topics for future work.

\bibliographystyle{IEEEtran}
\bibliography{IEEEabrv,Amirhossein}

\end{document}